\documentclass[aps,showpacs,longbibliography,notitlepage,superscriptaddress,twocolumn]{revtex4-1}
\pdfoutput=1
\usepackage[utf8]{inputenc}
\usepackage[english]{babel}
\usepackage[T1]{fontenc}
\usepackage{amsmath}
\usepackage{xcolor}
\colorlet{myPurple}{blue!40!red}
\colorlet{myPurplee}{blue!10!red}
\colorlet{myCyan}{cyan!60!gray}
\colorlet{myRed}{red!66!black}
\usepackage{tikz}
\usepackage{pgfplots}
\usepackage{graphicx}
\usepackage{subcaption}
\pgfplotsset{compat=1.14}
\usepackage[colorlinks=true,citecolor=myRed,urlcolor=myRed,linkcolor=myRed]{hyperref}
\usepackage[normalem]{ulem}
\usepackage{exscale}
\usepackage{bbm}
\usepackage{graphicx}
\usepackage{amsmath}
\usepackage{latexsym}
\usepackage{amsfonts}
\usepackage{amssymb}
\usepackage{times}
\usepackage[T1]{fontenc}
\usepackage{amsthm}
\usepackage{enumerate}
\usepackage{bbold}
\usepackage{color}
\usepackage{nicefrac}
\newcommand{\sket}[1]{{\ensuremath{\lvert#1\rangle}}}
\newcommand{\lket}[1]{{\ensuremath{\left\lvert#1\right\rangle}}}
\newcommand{\ket}[1]{\if@display\lket{#1}\else\sket{#1}\fi}
\newcommand{\tp}{\otimes}
\newcommand{\sbra}[1]{{\ensuremath{\langle#1\rvert}}}
\newcommand{\lbra}[1]{{\ensuremath{\left\langle#1\right\rvert}}}
\newcommand{\bra}[1]{\if@display\lbra{#1}\else\sbra{#1}\fi}

\newcommand{\sbraket}[2]{{\ensuremath{\langle#1\rvert#2\rangle}}}
\newcommand{\lbraket}[2]{{\ensuremath{\left\langle#1\!\left\rvert\vphantom{#1}#2\right.\!\right\rangle}}}
\newcommand{\braket}[2]{\if@display\lbraket{#1}{#2}\else\sbraket{#1}{#2}\fi}

\newcommand{\sketbra}[2]{{\ensuremath{\lvert #1\rangle\!\langle #2\rvert}}}
\newcommand{\lketbra}[2]{{\ensuremath{\left\lvert #1\right\rangle\!\!\left\langle #2\right\rvert}}}
\newcommand{\ketbra}[2]{\if@display\lketbra{#1}{#2}\else\sketbra{#1}{#2}\fi}
\usepackage{tcolorbox}
\usepackage{mathtools}

\newcommand{\proj}[1]{\ketbra{#1}{#1}}

\newcommand{\tr}{\textrm{Tr}}

\newcommand{\idd}{\mathds{1}}

\newcommand{\rA}{\text{A}}
\newcommand{\rB}{\text{B}}
\newcommand{\rC}{\text{C}}
\newcommand{\rD}{\text{D}}

\newcommand{\M}{\mathsf{M}}

\usepackage{tikz}
\newcommand{\qleq}{\quad\leq\quad}
\newcommand{\qgeq}{\quad\geq\quad}
\usepackage{lipsum}
\theoremstyle{plain}
\newtheorem{thm}{Theorem}
\newtheorem{lem}{Lemma}

\usepackage{graphicx}
\usepackage{bm}
\usepackage{dsfont}
\usepackage{tikz}
\usepackage[T1]{fontenc}
\usepackage{amsthm}
\usepackage{array}
\usepackage{amssymb}
\usepackage{amsfonts}
\usepackage{cancel}
\usepackage[toc,page]{appendix}
\usepackage{multirow}
\usepackage{color}
\usepackage{calrsfs}
\usetikzlibrary{backgrounds,decorations.pathreplacing,calc}

\usepackage{tkz-euclide}
\usetkzobj{all}

\newcommand{\la}{\gamma}

\newcommand{\lp}{\beta}
\newcommand{\ld}{\alpha}
\newcommand{\lt}{{\lambda}}

\begin{document}

\title{Quantum nonlocality in networks can be demonstrated with an arbitrarily small level of independence between the sources}

\author{Ivan \v{S}upi\'{c}}
\email{ivan.supic@unige.ch}
\author{Jean-Daniel Bancal}
\author{Nicolas Brunner}
\affiliation{Département de Physique Appliquée, Université de Genève, 1211 Genève, Switzerland}

\date{\today}

\begin{abstract}
Quantum nonlocality can be observed in networks even in the case where every party can only perform a single measurement, i.e. does not receive any input. So far, this effect has been demonstrated under the assumption that all sources in the network are fully independent from each other. Here we investigate to what extent this independence assumption can be relaxed. After formalizing the question, we show that, in the triangle network without inputs, quantum nonlocality can be observed, even when assuming only an arbitrarily small level of independence between the sources. This means that quantum predictions cannot be reproduced by a local model unless the three sources can be perfectly correlated. 
\end{abstract}

\maketitle

\section{Introduction}

One of the distinctive features of quantum theory is that it allows for existence of nonlocal correlations. Distant observers sharing an entangled state can generate strong correlations by performing well-chosen local measurements. As shown by Bell \cite{Bell}, these correlations are in fact so strong that they cannot be reproduced in any physical theory consistent with a natural definition of locality. 
Besides being a thought-provoking feature scrutinized by physicists and philosophers alike, Bell nonlocality is tightly connected with the emergence of quantum technologies \cite{review}. The whole concept of device-independent quantum information processing relies on the concept of Bell nonlocality as a resource \cite{diqkd,pironio2010random,colbeck}. Also, quantum nonlocal correlations inspired the first proof of unconditional quantum computational advantage \cite{Bravyi_2018}. 

Recently, the concept of quantum nonlocality has been investigated in a novel setting, namely from the perspective of quantum networks. The latter consist of some distant parties (nodes), which are interconnected via a number of separated quantum sources. Typically each source distributes entanglement only to a subset of parties. Each node can then jointly process or measure quantum systems originating from different sources (e.g. via entangled measurements) which may generate strong correlations across the entire network, in the spirit of quantum teleportation or entanglement swapping. 

The network configuration allows for a number of remarkable phenomena, radically different from the usual Bell scenario. Most remarkably, it is here possible to observe quantum nonlocality without inputs, i.e. by having each observer perform a single (fixed) quantum measurement \cite{Fritz_2012,Branciard_2012,fraser,Gisin_2019,Renou_2019,krivchy2019neural}. Moreover, networks can reveal the nonlocality of certain entangled states, which could not lead to nonlocality in the usual Bell scenario \cite{Cavalcanti_2011,Cavalcanti_2012}. Finally, quantum networks may allow for completely new forms of quantum nonlocality, via the adequate combination of entangled states and entangled measurements \cite{Renou_2019}.  

The above phenomena rely on a fundamental assumption, namely the fact that all sources in the network are fully independent from each other. In particular, the formulation of the concept of locality (and hence the definition of nonlocality) in networks is based on this assumption. This idea is known as ``$N$-locality'', and represents a natural generalization of Bell locality to networks \cite{Branciard_2010,Fritz_2012}. 

It is natural, however, to challenge this assumption of fully independent sources. From a conceptual point of view, one may want to find what are the minimal requirements (in particular in terms of assumptions) for demonstrating quantum nonlocality in networks. From a more practical point of view, it is natural to consider that the various sources in a quantum network could be classically correlated to some extent; for instance, the sources may have to be calibrated and/or synchronized. 

This is precisely the question we investigate in this work: can one observe quantum nonlocality in networks when relaxing the hypothesis of full independence of the sources? If yes, what is the minimal level of independence required. This question is indeed of particular interest when considering network Bell tests where the parties receive no input, i.e. perform a single (fixed) measurement. 
 
We answer the first question in the affirmative. To do so, we first develop a framework to formalize the problem. We then show that an arbitrarily small level of independence of the sources is enough for observing quantum nonlocality in networks, even when the parties receive no input. This represents our main result, which we prove for the so-called triangle network without inputs (see Fig. 1). We construct a family of quantum experiments (involving three fully independent quantum sources), and show that the resulting correlations can provably not be simulated classically unless the three sources can be maximally correlated. This implies that no classical model where the sources are partially (but not fully) correlated can reproduce the quantum predictions. Hence quantum nonlocality can be observed with only an arbitrarily small level of independence between the sources. In the last part of the paper, we extend our analysis to a quantum nonlocal distribution recently proposed by Renou et al. \cite{Renou_2019}, as well as to the so-called bilocality network \cite{Branciard_2010,biloc2}, which we recast as a square network without inputs.

\begin{figure*}
\begin{subfigure}{.49\textwidth}
  \centering
  \includegraphics[width=.49\linewidth]{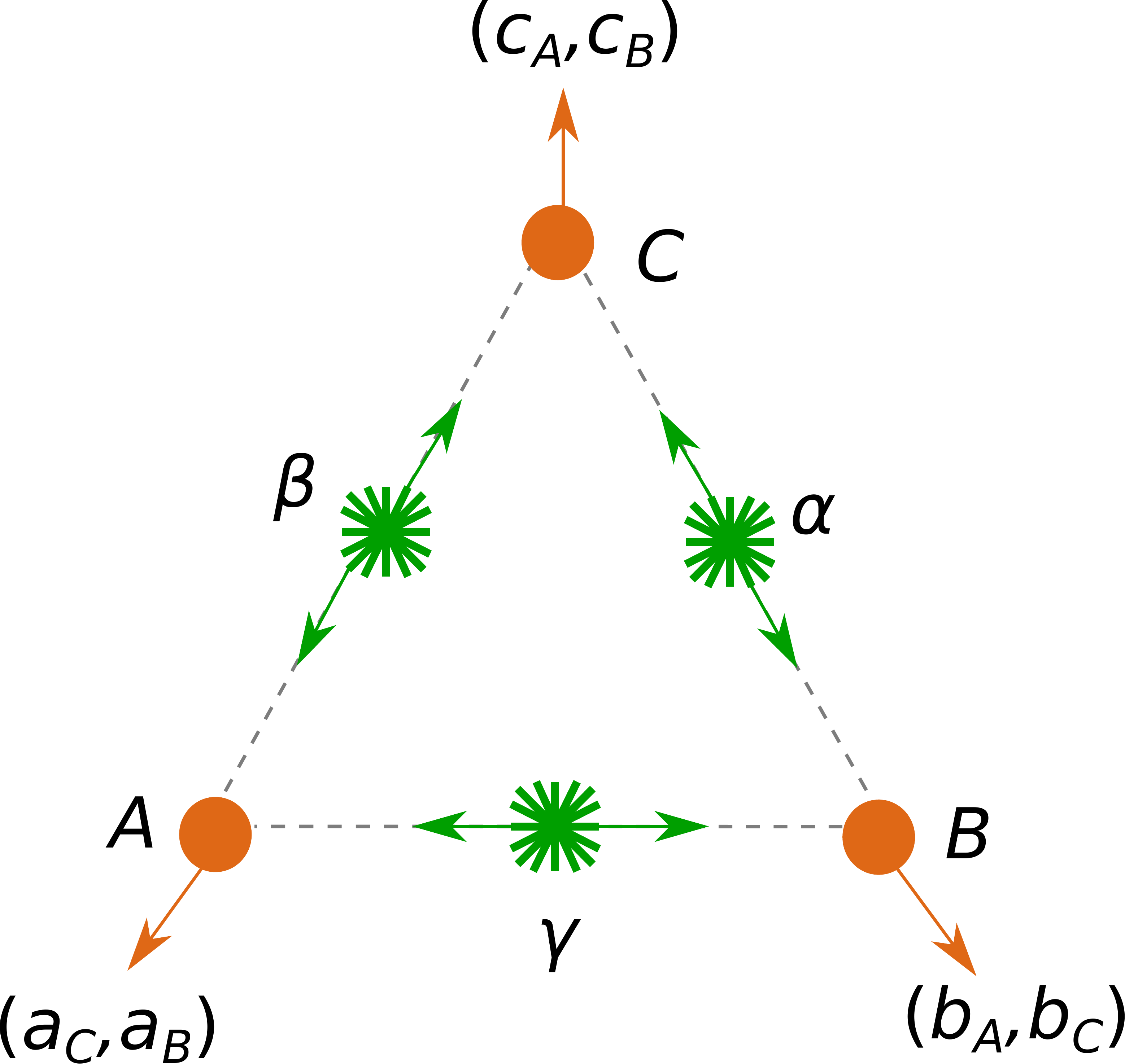}  
  \label{fig:IndependentTriangle}
\end{subfigure}
\begin{subfigure}{.49\textwidth}
  \centering
  \includegraphics[width=.49\linewidth]{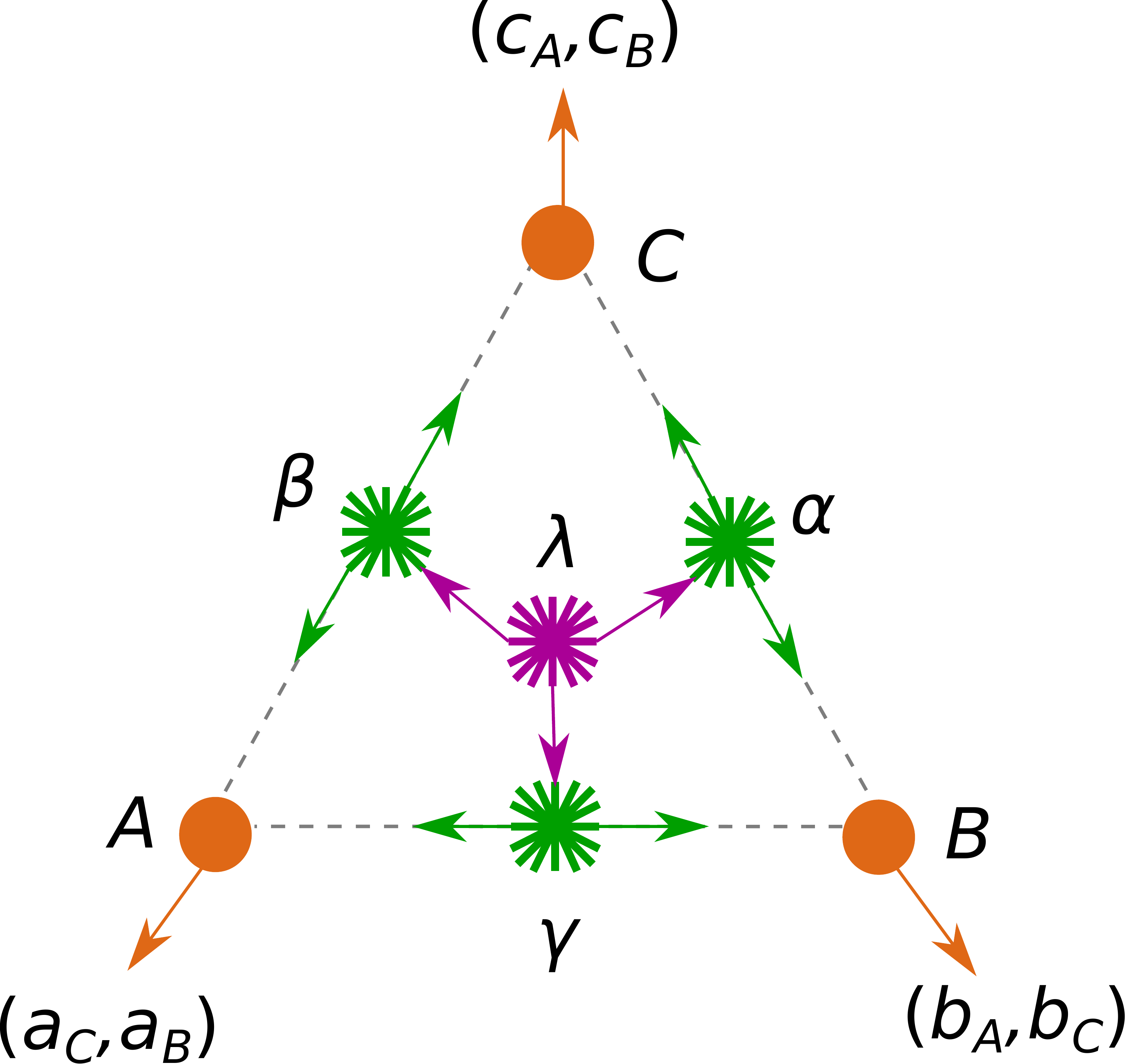}  
  \label{fig:OverlappingTriangle}
\end{subfigure}
\caption{(Left) The triangle network without inputs, featuring three parties connected by three bipartite sources. Here the three sources are assumed to be fully independent from each other. Previous works have demonstrated quantum nonlocality in this setting \cite{Fritz_2012,fraser,Renou_2019}. (Right) In this work, we consider a more general network where the three sources can become correlated, via an additional central source. Our main result is that quantum nonlocality can still be observed, even when assuming only an arbitrarily small level of independence of the three sources.}
\label{fig:IndVSOverlap}
\end{figure*}

\section{Setting}

Most of our analysis focuses on the so-called triangle network without inputs. We consider a quantum experiment in such a network, involving three parties (Alice, Bob and Charlie) and three separate sources. Each source produces a bipartite quantum state, which is distributed to every pair of parties: Alice and Bob share the state $\varrho^{\rA_\rB\rB_\rA}$, Alice and Charlie $\varrho^{\rA_\rC\rC_\rA}$ and Bob and Charlie $\varrho^{\rB_\rC\rC_\rB}$. Next each party performs a local (possibly joint) measurement on their two local subsystems: Alice performs the measurement $\M_a^{\rA_\rB\rA_\rC}$ giving outcome $a$, while Bob and Charlie apply measurements $\M_b^{\rB_\rA\rB_\rC}$ and $\M_c^{\rC_\rB\rC_\rA}$, with outputs $b$ and $c$ respectively. The resulting statistics is given by
\begin{multline} \label{quantum}
  p_Q(a,b,c) = \tr\big[\big(\M_a^{\rA_\rB\rA_\rC}\tp \M_b^{\rB_\rA\rB_\rC}\tp \M_c^{\rC_\rA\rC_\rB}\big)\times \\ \times \big(\varrho^{\rA_\rB\rB_\rA}\tp\varrho^{\rA_\rC\rC_\rA}\tp\varrho^{\rB_\rC\rC_\rB}\big)\big] \,.
\end{multline}
Note that when calculating the above expression, one should be careful about the order of the various subsystems. 

It turns out that there exist such quantum distributions which exhibit nonlocality, despite the fact that each party uses a fixed measurement \cite{Fritz_2012,fraser,Renou_2019}. More formally, this means that $p_Q$ does not admit a decomposition of the following form: 

\begin{equation}\label{trianglelocal}
  p(a,b,c) = \sum_{\alpha,\beta,\gamma} p(\alpha)p(\beta)p(\gamma)  p(a|\beta,\gamma)p(b|\alpha,\gamma)p(c|\alpha,\beta),
\end{equation}
where $\alpha$, $\beta$, $\gamma$ denote the classical variables distributed by the sources. This form makes the assumption of $N$-locality apparent, namely that all three sources are fully independent from each other, i.e. $p(\alpha,\beta,\gamma) =p(\alpha)p(\beta)p(\gamma) $. A distribution of the form of \eqref{trianglelocal} is said to be trilocal. Note that determining whether a given distribution is trilocal remains a challenging problem, despite recent progress \cite{krivchy2019neural,navascues2017inflation,wolfe2019quantum,Pozas_Kerstjens_2019}.

In this work, we investigate a more general class of local models where the three sources can be correlated to some degree. To do so, we introduce an additional central source, distributing a classical variable $\lambda$ to all three sources (see Fig. \ref{fig:IndVSOverlap}). This new variable may influence the choice of variables $\alpha$, $\beta$ and $\gamma$, thus introducing correlations between them. Clearly, this causal influence cannot be unrestricted, otherwise any possible distribution $p(a,b,c)$ can be reproduced simply by setting $\lambda= (a,b,c)$, sampled from the distribution $p(a,b,c)$, and having $\alpha=\beta=\gamma=\lambda$. In order to quantitatively limit the causal influence, we introduce the condition

\begin{equation}\label{condition}
     p(\alpha,\beta,\gamma|\lambda) \geq \varepsilon_1p(\alpha)p(\beta)p(\gamma), \qquad \forall \alpha,\beta,\gamma,\lambda 
\end{equation}
where $\varepsilon_1\in (0,1]$ is a constant. This implies that 
the following condition holds
\begin{equation}\label{condition2}
     p(\alpha,\beta,\gamma|\lambda) \leq \varepsilon_2(\alpha,\beta,\gamma)p(\alpha)p(\beta)p(\gamma), \qquad \forall \alpha,\beta,\gamma,\lambda 
\end{equation}
where $\varepsilon_2(\alpha,\beta,\gamma) \in [1,1/[p(\alpha)p(\beta)p(\gamma)])$. Condition \eqref{condition} implies that there is no value  $\lambda$, for which $p(\alpha,\beta,\gamma) = 0$, which further implies that $\Lambda$ never deterministically causes the values of $\Lambda_A,\Lambda_B$ and $\Lambda_C$, as expressed in \eqref{condition2}. We say that a distribution $p(a,b,c)$ is $\varepsilon_1$-trilocal if it admits a decomposition of the form
\begin{equation}\label{triangleDSlocal}
  p(a,b,c) = \sum_{{\lambda},\alpha,\beta,\gamma} p({\lambda})p(\alpha,\beta,\gamma|{\lambda})  p(a|\beta,\gamma)p(b|\alpha,\gamma)p(c|\alpha,\beta),
\end{equation}
where $p(\alpha,\beta,\gamma|{\lambda})$ satisfies both constraints \eqref{condition} and \eqref{condition2}. The values assigned to the parameter $\varepsilon_1$ quantify the degree of independence. Setting $\varepsilon_1 = 1$, we recover the usual definition of trilocality of Eq. \eqref{trianglelocal},  i.e. the scenario with fully independent sources. In the opposite limit, imposing only $\varepsilon_1>0$ represents the regime with an arbitrarily small level of independence. Note that condition \eqref{condition} implies that even when knowing the triple $(\beta,\gamma,\lambda)$ there is still some uncertainty about the value of $\alpha$, i.e. all values $\alpha$ are still possible. If this was not the case, and the triple $(\beta,\gamma,\lambda)$ would imply a certain value $\alpha$, for all other values $\alpha'$ it would hold that $p(\alpha',\beta,\gamma|\lambda) = 0$, which contradicts condition \eqref{condition}. A similar argument holds for impossibility to perfectly predict $\beta$ or $\gamma$ given the value of other three hidden variables.

\section{Main result}

We now show that there exist quantum correlations which cannot be explained by any local model, even when the sources features an arbitrary small level of independence. This represents our main result. Specifically, we consider the triangle network without inputs, as discussed in the previous section. The output of each party consists of two classical bits, which we denote $a = (a_B,a_C)$, $b= (b_A,b_C)$ and $c = (c_A,c_B)$ (see Fig. \ref{fig:IndVSOverlap}). We construct a quantum distribution of the form of \eqref{quantum}, hence featuring three fully independent quantum sources. Then we prove that this quantum distribution is not $\varepsilon_1$-trilocal for any $\varepsilon_1>0$, i.e. it does not admit a decomposition of the form of
\eqref{triangleDSlocal} for any $\varepsilon_1>0$. Hence it cannot be reproduced by any trilocal model unless the three sources are fully correlated.

The construction of the quantum distribution is as follows. Alice and Bob share the partially entangled pair of qubits 
\[\varrho^{\rA_\rB\rB_\rA} = \ketbra{\psi_\theta}{\psi_\theta},\]
where 
\begin{equation}
\ket{\psi_\theta} = \left(\sin\theta\ket{11} + \cos{\theta}\left(\ket{01}+\ket{10}\right)\right)/\sqrt{1+\cos^2{\theta}}.
\end{equation}
Charlie shares with both  Alice and Bob, a classically correlated state of the form \[\varrho^{\rA_\rC\rC_\rA} = \varrho^{\rB_\rC\rC_\rB} = \frac{1}{2}(\ketbra{00}{00} + \ketbra{11}{11}).\] Charlie measures both qubits in the computational basis and outputs the results. Similarly, Alice and Bob measure their qubit shared with Charlie in the computational basis, i.e.
\begin{align*}
\M_c^{\rC_\rB\rC_\rA} &= \ketbra{c_Ac_B}{c_Ac_B},\\
\M^{\rA_\rC}_{a_C} = \ketbra{a_C}{a_C} &\quad \M^{\rB_\rC}_{b_C} = \ketbra{b_C}{b_C}
\end{align*}
This way, Alice's output bit $a_C$ (Bob's bit $b_C$) is perfectly correlated with Charlie's output bit $c_A$ ($c_B$). The choice of the measurements $\M^{\rA_\rB}_{a_B}$ and $\M^{\rB_\rA}_{b_A}$ are conditioned on the value of bits $a_C$ and $b_C$. These measurements are performed on the shared entangled state $\varrho^{\rA_\rB\rB_\rA}$. If $a_C=0$ then $\M^{\rA_\rB}_{a_B} = \ketbra{a_B}{a_B}$, and analogously for Bob, $b_C = 0$ implies $\M^{\rB_\rA}_{b_A} = \ketbra{b_A}{b_A}$. If $a_C = 1$, the measurement basis is given by
$\M_{a_B}^{\rA_\rB} = \ketbra{w_{a_B}}{w_{a_B}}$, where
\begin{align*}
    \ket{w_0} &= \sin\theta\ket{0} + \cos\theta\ket{1},\\
    \ket{w_1} &= \cos{\theta}\ket{0} - \sin{\theta}\ket{1}.
\end{align*}
Similarly for Bob, when $b_C = 1$, he uses the same local measurement, i.e.  $\M_{b_A}^{\rB_\rA} = \ketbra{w_{b_A}}{w_{b_A}}$. 

The resulting probability distribution is obtained from \eqref{quantum}, and denoted by $p_Q^\theta(a,b,c) $. Note that this represents a family of distribution parametrized by $\theta$. For $\theta=\pi/4$ we recover the nonlocal distribution proposed by Fritz \cite{Fritz_2012}, while for $\theta=0$ the distribution is local, i.e. admits a decomposition of the form \eqref{trianglelocal}.

%
%

Now we are ready to state the main result:
\begin{thm}\label{theorem}
The quantum distribution $p_Q^\theta(a,b,c) $ 
cannot be reproduced by a local model of the form \eqref{triangleDSlocal} for any $\varepsilon_1>0$, when $\theta \in(0,\pi/4) $.
\end{thm}

Before sketching the formal proof, we give some intuition. The quantum distribution $p_Q^\theta(a,b,c) $ generalises the construction of Fritz \cite{Fritz_2012}. The main idea is to embed a standard bipartite Bell test (between Alice and Bob) in the triangle network configuration. In this way, the inputs for Alice and Bob, which are necessary in the standard Bell test, can be effectively replaced by the outputs $a_C$ and $b_C$. By verifying the condition
\begin{equation}\label{consistency}
    p(A_C = C_A \wedge B_C = C_B) = 1,
\end{equation}
one can ensure that the effective inputs $a_C$ and $b_C$ were generated by the sources $\beta$ and $\alpha$, respectively. When assuming fully independent sources, as in \cite{Fritz_2012}, this ensures that the effective inputs are chosen independently from the $\gamma$ source. Hence, when the conditional distribution $p(a_B,b_A|a_C,b_C)$ is nonlocal in the standard Bell scenario (witnessed, for instance, via violation of the CHSH Bell inequality) this implies that the full quantum distribution $p_Q^\theta(a,b,c) $ is nonlocal in the triangle network, i.e. cannot be decomposed as \eqref{trianglelocal}.

In our case, we would like to go one step further, in the sense that the three sources may now become correlated via the central source $\lambda$. Hence the effective inputs $a_C$ and $b_C$ may now become correlated with the $\gamma$ source. However, as long as the condition \eqref{consistency} holds, it follows from the conditions \eqref{condition} and \eqref{condition2} that the effective inputs and the $\gamma$ source are not perfectly correlated. We can then make use of a result by P\"utz et al. \cite{Puetz_2014,Puetz_2016}, proving that quantum nonlocality can be observed in the standard Bell scenario for an arbitrarily small level of measurement independence. Finally, by embeding this quantum distribution in the triangle network, we obtain the quantum distribution $p_Q^\theta(a,b,c) $ of Theorem 1.

\begin{proof} 
The proof works by constructing explicitly  a Bell inequality for the triangle network, satisfied by all $\varepsilon_1$-trilocal distributions with $\varepsilon_1>0$ under the condition \eqref{consistency}. One can then check that the quantum distribution $p_Q^\theta(a,b,c) $ violates it for any $\varepsilon_1>0$, as long as $0<\theta<\pi/4 $.

To construct the desired Bell inequality, we start from a standard bipartite Bell inequality of the form
 %
 \begin{multline}\label{ineqFr}
     \mathcal{I}_{Bell}(p(a_B,b_A|a_C,b_C)) = \\ = \sum_{a,b}\left(\omega_{a,b}^+p(a_B,b_A|a_C,b_C) - \omega_{a,b}^-p(a_B,b_A|a_C,b_C)\right) \leq L \,,
 \end{multline}
where variables $a_C$ and $b_C$ play the role of inputs in the conditional probabilities. Here, $L$ denotes the local bound, and all $\omega_{a,b}^+$ and $\omega_{a,b}^-$ are real non-negative coefficients. Any bipartite Bell inequality can be written in this form~\cite{review}.

We can now obtain a Bell inequality for the triangle network. Specifically, for all $\varepsilon_1$-trilocal distributions of the form \eqref{triangleDSlocal} satisfying the condition \eqref{condition}, the following inequality must be satisfied: 
\begin{multline}\label{Bell}
     I \equiv \mathcal{I}_{Bell}^{so}(p(a_B,b_A,a_C,b_C),\xi_1,\xi_2) = \\ = \sum_{a,b}\left(\xi_1\omega_{a,b}^+p(a_B,b_A,a_C,b_C) - \xi_2\omega_{a,b}^-p(a_B,b_A,a_C,b_C)\right) \\
    \leq \xi_1\xi_2 L \,, 
 \end{multline}
where $\xi_1$ and $\xi_2$ are positive numbers specified below. The main ingredient for proving this Bell inequality is the following lemma:
\begin{lem}\label{lem:inequalities}
For all $\varepsilon_1$-trilocal distributions with $\varepsilon_1>0$ satisfying Eq.~\eqref{consistency}, there exist $\xi_1 > 0$ and $\xi_2 < 1/p(a_B,b_A,a_C,b_C)$, such that
the following two inequalities hold:
\begin{align}\label{Mselfisolate1}
    p(a_B,b_A,a_C,b_C)  &\geq 
\xi_1 \sum_{\gamma}p(\gamma)p(a_B|a_C,\gamma)p(b_A|b_C,\gamma)\\ \label{Mselfisolate2}
    p(a_B,b_A,a_C,b_C)  &\leq 
\xi_2 \sum_{\gamma}p(\gamma)p(a_B|a_C,\gamma)p(b_A|b_C,\gamma),
\end{align}
for all $a_B, b_A$, $a_C$ with $p(a_C)>0$, and $b_C$ with $p(b_C)>0$.
The parameters $\xi_1$ and $\xi_2$ are defined in the following way
\begin{align}
\xi_1 &= \frac{\varepsilon_1^3}{\varepsilon_2^6}\min_{a_C:p(a_C)>0}p(a_C)\min_{b_C:p(b_C)>0}p(b_C)\\
\xi_2 &= \frac{\varepsilon_2^3}{\varepsilon_1^6}\max_{a_C:p(a_C)>0}p(a_C)\max_{b_C:p(b_C)>0}p(b_C),
\end{align}
where $\varepsilon_2 = \max_{\alpha,\beta,\gamma,\lambda}\varepsilon_2(\alpha,\beta,\gamma,\lambda)$.
\end{lem}
The proof of the lemma is presented in Appendix \ref{app:lemma}. The Bell inequality \eqref{Bell} is obtained by applying the conditions \eqref{Mselfisolate1} and \eqref{Mselfisolate2} to the Bell expression $I$, and using the fact that $\mathcal{I}_{Bell}(p(a_B,b_A|a_C,b_C)) \leq L$ is a standard Bell inequality.

In order to prove the theorem, we now apply the above construction starting from an appropriate bipartite Bell inequality  $\mathcal{I}_{Bell}$. Similarly to Ref. \cite{Puetz_2014}, we start from the well-known Clauser-Horn-Shimony-Holt (CHSH) Bell inequality, expressed in the form  introduced in Ref. \cite{Eberhardt}:
\begin{equation}\label{eberhardt}
    p(00|00) - (p(01|01)+p(10|10)+p(00|11)) \leq 0.
\end{equation}
Using the above method, we arrive at the Bell inequality 
\begin{equation}
\begin{split}
    \xi_1 p(&A_B=0,B_A=0,A_C=0,B_C=0) - \\ - \xi_2(&p(A_B=0,B_A=1,A_C=0,B_C = 1)+ \\ + &p(A_B=1,B_A=0,A_C=1,B_C=0)+ \\ + &p(A_B=0,B_A=0,A_C=1,B_C=1)) \leq 0 \,.
\end{split}
\end{equation}
By construction this Bell inequality holds for all $\varepsilon_1$-trilocal models given that condition \eqref{consistency} is satisfied. Note that the condition $\varepsilon_1>0$ implies that $\xi_1>0$ and $\xi_2>0$.  

Finally, we need to show that the quantum distribution $p_Q^\theta(a,b,c) $ violates the above Bell inequality for any $\varepsilon_1>0$, when $0<\theta<\pi/4 $. This can be straightforwardly seen by noticing that the quantum distribution has the following properties: 
\begin{equation}
\begin{split}
    p_Q^{\theta}(A_B=0,B_A=0,A_C=0,B_C=0) &>0  \\  p_Q^{\theta}(A_B=0,B_A=1,A_C=0,B_C = 1)&=0  \\ p_Q^{\theta}(A_B=1,B_A=0,A_C=1,B_C=0) &=0  \\ p_Q^{\theta}(A_B=0,B_A=0,A_C=1,B_C=1) &= 0 \,,
\end{split}
\end{equation}
which concludes the proof. Note that the above properties of $p_Q^\theta(a,b,c) $ correspond to Hardy's paradox \cite{Hardy_1993}.
\end{proof}

\section{Further cases}

The techniques developed above can be used to discuss other instances of quantum nonlocality in networks, in particular beyond the triangle network. There are however not so many examples of such distributions known so far, in particular for networks without inputs. Here we focus on two case studies.

We first consider the quantum distribution presented in Ref. \cite{Renou_2019}, which we refer to as the RBBBGB distribution. While it is also constructed for the triangle network without inputs, the RBBBGB setup has a different structure compared to the quantum distributions discussed above. In particular, all three sources produce an entangled state, and all parties perform the same measurements. Importantly, this measurement is entangled, i.e. some of its eigenstates are entangled. Here we can prove that the RBBBGB distribution remains nonlocal even if the sources are correlated to some extent. Specifically, we derive a bound on the parameter $\varepsilon_1$ such that the distribution is still nonlocal. However, contrary to the result of the previous section, nonlocality cannot be guaranteed here for any $\varepsilon_1 > 0$. All details can be found in Appendix \ref{rbbbgbproof}.

The second example considers the square network without inputs. This network features four parties and four sources, each source connecting a pair of neighbouring parties. The starting point is the simpler network of entanglement swapping (known as the ``bilocality network'' \cite{Branciard_2010,Branciard_2012}). In this case, however, it is necessary to have inputs in order to obtain quantum nonlocality. These inputs can be effectively removed by involving an additional party, similarly to the construction discussed in the previous section as well as in \cite{Fritz_2012}. Here we show that quantum nonlocality can still be observed, even when assuming partial correlations between the four sources. All details can be found in Appendix \ref{app:square}.

\section{Discussion}

We have discussed quantum nonlocality in networks where the parties receive no input. While this phenomenon has been demonstrated so far only under the assumption of fully independent sources, we have shown here that this assumption can be relaxed. In fact, an arbitrarily small level of independence is enough for demonstrating quantum nonlocality. 

An interesting question is how the geometry of the network affects our result. That is, while we have mostly focused here on the triangle network, could a similar result hold for all networks that feature quantum nonlocality. It would also be interesting to consider different ways of correlating the sources, for instance by correlating only two out of the three sources in Fig. 1. Another direction is to consider different ways of quantifying correlations between the sources, using e.g. entropic quantities such as mutual information.

\begin{acknowledgements}
  This work was supported by the Swiss National Science Foundation (Starting Grant DIAQ and NCCR SwissMap).
\end{acknowledgements}

\bibliographystyle{unsrturl}
\bibliography{biblio}

\onecolumngrid
\appendix

\appendixpage
\addappheadtotoc

\section{Proof of Lemma \ref{lem:inequalities}}\label{app:lemma}

In this appendix we give the proof of Lemma \ref{lem:inequalities}. Let us repeat some basic equations important for our model. We say that a behaviour $p(a,b,c)$ is $\varepsilon_1$-trilocal if it admits the following decomposition
\begin{equation}\label{triangleDSlocalSM}
  p(a,b,c) = \sum_{{\lambda},\alpha,\beta,\gamma} p({\lambda})p(\alpha,\beta,\gamma|{\lambda})  p(a|\beta,\gamma)p(b|\alpha,\gamma)p(c|\alpha,\beta)
\end{equation}
and the hidden variables satisfy the relaxed independence condition
\begin{equation}\label{independenceCondition}
 p(\alpha)p(\beta)p(\gamma)\epsilon_1 \leq  p(\alpha,\beta,\gamma|{\lambda})) \leq p(\alpha)p(\beta)p(\gamma)\epsilon_2({\lambda},\alpha,\beta,\gamma) \qquad \forall {\lambda},\alpha,\beta,\gamma,
\end{equation}
where $\epsilon_1 \in (0,1]$ and $\epsilon_2(\lambda,\alpha,\beta,\gamma) \in [1,1/p(\alpha)p(\beta)p(\gamma))$. Condition~\eqref{independenceCondition} excludes deterministic relation between any triple $(\alpha,\beta,\gamma)$ and any ${\lambda}$. This means that every triple $(\ld,\lp,\la)$ appears with nonzero probability independently of which value $\lt$ the hidden variable $\Lambda$ takes. 
Hence, by introducing new variables $\eta$, the following inequalities must hold for marginal distributions as well:
\begin{align}\nonumber
 p(\alpha)p(\beta)\varepsilon_1 &\leq  p(\alpha,\beta|{\lambda}) \leq p(\alpha)p(\beta)\varepsilon_2\qquad \forall \lambda,\ld,\lp,\\ \label{independenceCondition2}
 p(\lp)\varepsilon_1 &\leq p(\beta|\lt,\la) \leq p(\lp)\varepsilon_2\qquad \forall \lt,\beta,\gamma\\ \nonumber
 p(\ld)\varepsilon_1 &\leq p(\alpha|\lt,\la) \leq p(\ld)\varepsilon_2\qquad \forall \lt,\alpha,\gamma \\ \nonumber
 p(\alpha)p(\beta)\varepsilon_1 &\leq  p(\alpha,\beta) \leq p(\alpha)p(\beta)\varepsilon_2\qquad \forall \ld,\lp,
 \end{align}
where $\varepsilon_2 = \max_{\alpha,\beta,\gamma,\lambda}$. The first set of inequalities is obtained from Eq.~\ref{independenceCondition} by summing over $\gamma$. The second is obtained from Eq.~\ref{independenceCondition} by summing over $\alpha$ and $\gamma$, and noting that $p(\beta|\lambda,\gamma) = p(\beta|\lambda)$. The analogous reasoning holds for the third set of inequalities. The last set of inequalities is obtain from Eq.~\ref{independenceCondition} by multiplying with $p(\lambda)$ and summing over $\lambda$. 

Another assumption in lemma \ref{lem:inequalities} is the consistency between certain outputs:
\begin{equation}\label{fritzdistributionGen}
    p(C_A = A_C \wedge C_B = B_C) = 1. 
\end{equation}
Such behaviour satisfies
\begin{align}\label{eq:oneStep}
    p(a_B,A_C=a_C,b_A,B_C=b_C|\lambda,\ld,\lp,\gamma) &= \sum_{c_Ac_B}p(a_B,A_C=a_C,b_A,B_C=b_C,C_A=c_A,C_B=c_B|\lambda,\ld,\lp,\gamma)\\
    &= p(a_B, A_C=C_A=a_C, b_A, B_C=C_B=b_C|\lambda,\ld,\lp,\gamma) 
\end{align}
where the first equality comes from the marginal distribution of  $p(a_B,a_C,b_A,b_C,c_A,c_B|\lambda,\beta,\alpha,\lambda)$, while the second is a direct consequence of \eqref{fritzdistributionGen}. Here, we use capital letters ($A_C$) for random variables and small letters ($a_C$) for the value they take. When clear from the context, we may simply use a small letter $a_C$ to refer to the condition $A_C=a_C$, hence not explicitly repeating the name of the associated random variable. Similarly to~\eqref{eq:oneStep}, we also have equalities:
\begin{align}
    p(a_B,b_A,C_A=c_A,C_B=c_B|\lambda,\ld,\lp,\gamma)  &=  \sum_{a_C,b_C}p(a_B,A_C=a_C,b_A,B_C=b_C,C_A=c_A,C_B=c_B|\lambda,\ld,\lp,\gamma)\\
    &= p(a_B,A_C=C_A=c_A,b_A,B_C=C_B=c_B|\lambda,\ld,\lp,\gamma).
\end{align}
These equalities imply
\begin{equation}\label{identityMoreGeneral}
    p(a_B,A_C=a_C,b_A,B_C=b_C|\lambda,\ld,\lp,\gamma) = p(a_B,b_A,C_A=a_C,C_B=b_C|\lambda,\ld,\lp,\gamma) 
    \qquad \forall a_B,a_C,b_A,b_C,\lambda,\ld,\lp,\gamma,
\end{equation}
and in particular
\begin{equation}\label{identity}
    p(A_C=a_C,B_C=b_C|\lambda,\ld,\lp,\gamma) = p(C_A=a_C,C_B=b_C|\lambda,\ld,\lp,\gamma) 
    \qquad \forall a_C,b_C,\lambda,\ld,\lp,\gamma.
\end{equation}\\

Let us consider the behaviour $p(a_B,a_C,b_A,b_C)$. Using relation~\eqref{identity} we can write it as
\begin{align}
    p(a_B,a_C,b_A,b_C) &=  \sum_{\lt,\ld,\lp,\la} p(\lt)p(\la|\lt)p(\lp,\ld|\lt,\la)p(A_C=a_C,B_C=b_C|\lambda,\ld,\lp,\la)p(a_B|\lambda,\lp,\la)(b_A|\lambda,\ld,\la)\\ 
    &=  \sum_{\lt,\ld,\lp,\la} p(\lt)p(\la|\lt)p(\lp,\ld|\lt,\la)p(C_A=a_C,C_B=b_C|\lambda,\ld,\lp,\la)p(a_B|\lambda,\lp,\la)(b_A|\lambda,\ld,\la)\\ 
     &=  \sum_{\lt,\ld,\lp,\la} p(\lt)p(\la|\lt)p(C_A=a_c,C_B=b_c,\ld,\lp|\lt,\la)p(a_B|\lt,\lp,\la)(b_A|\lt,\ld,\la)\\ 
    &=  \sum_{\lt,\ld,\lp,\la} p(\lt)p(\la|\lt)p(C_A=a_c,C_B=b_c,\ld,\lp|\lt)p(a_B|\lt,\lp,\la)(b_A|\lt,\ld,\la).
\end{align}
 To get the last equation we used the fact that $c,\alpha,\beta$, conditioned on $\lambda$ do not depend on $\gamma$.
 
 Without loss of generality we can assume that the local model is deterministic (and put all randomness in the distribution of the hidden variables). That means that certain values of the pair $(\ld,\lp)$ "induce" the pair $(c_A,c_B)$, i.e. for certain values $(\ld,\lp)$ it either holds that $p(c_A,c_B,\ld,\lp|\lt) = p(\ld,\lp|\ld,\lt)$ or $p(c_A,c_B,\ld,\lp|\lt) = 0$. Let us denote all the values of the pair $(\ld,\lp)$ that "induce" the pair $(c_A,c_B)$ with $\Omega_{c_A,c_B}$. Thus, the statements $(\lp,\ld) \in \Omega_{c_A,c_B}$ and $C = (c_A,c_B)$ are equivalent. For each pair $(C_A=a_c,C_B=b_c)$ in the last equation the contributions in the sums over $\lp$ and $\ld$ come from the values of $(\ld,\lp) \in \Omega_{a_C,b_C}$. Hence, the last expression reduces to
\begin{align}
    p(a_B,a_C,b_A,b_C) &=    \sum_{\substack{\lt,\la,\\ (\ld,\lp)\in \Omega_{a_C,b_C}}} p(\lt)p(\la|\lt)p(\ld,\lp|\lt)p(a_B|\lt,\lp,\la)p(b_A|\lt,\ld,\la).
\end{align} 
Given \eqref{independenceCondition2}, we can put a lower bound on this expression:
\begin{align}
p(a_B,a_C,b_A,b_C) &\quad \geq \quad
\varepsilon_1 \sum_{\substack{\lt,\la,\\
\ld,\lp \in \Omega_{a_C,b_C}}}p(\lt)p(\la|\lt)p(\ld)p(\lp)p(a_B|\lt,\lp,\la)p(b_A|\lt,\ld,\la)\\
&\quad \geq \quad
\frac{\varepsilon_1}{\varepsilon_2^2} \sum_{\substack{\lt,\la,\\
\ld,\lp \in \Omega_{a_C,b_C}}}p(\lt)p(\la|\lt)p(\ld|\lt,\la)p(\lp|\lt,\la)p(a_B|\lt,\lp,\la)p(b_A|\lt,\ld,\la)\\
&\quad = \quad
\frac{\varepsilon_1}{\varepsilon_2^2} \sum_{\substack{\lt,\la,\\
\ld,\lp \in \Omega_{a_C,b_C}}}p(\lt)p(\la|\lt)p(a_B,\lp|\lt,\la)p(b_A,\ld|\lt,\la)\\
&\quad = \quad
\frac{\varepsilon_1}{\varepsilon_2^2} \sum_{\substack{\lt,\la,\\
\ld,\lp \in \Omega_{a_C,b_C}}}p(\lt)p(\la|\lt)p(a_B,b_A,\ld,\lp|\lt,\la), \label{eq:interruption1}
\end{align}
where we used the causal structure given in Fig~\ref{fig:IndVSOverlap}b. Now, given the definition of $\Omega_{c_A,c_B}$, we can rewrite Eq.~\eqref{eq:interruption1} as
\begin{align}
p(a_B,a_C,b_A,b_C) &\quad \geq \quad \frac{\varepsilon_1}{\varepsilon_2^2} \sum_{\lt,\la}p(\lt)p(\la|\lt)p(a_B,b_A,C_A=a_C,C_B=b_C|\lt,\la)
\\
&\quad = \quad
\frac{\varepsilon_1}{\varepsilon_2^2} \sum_{\lt,\la}p(\lt)p(\la|\lt)p(a_B,A_C=a_C,b_A,B_C=b_C|\lt,\la) \label{eq:interruption2},
\end{align}
where we used Eq.~\eqref{identityMoreGeneral} multiplied by $p(\alpha,\beta|\lambda,\gamma)$ and summed over $\alpha$ and $\beta$. Finally, using again the causal structure of Fig.~\ref{fig:IndVSOverlap}b we can rewrite Eq.~\eqref{eq:interruption2} as
\begin{align}
p(a_B,a_C,b_A,b_C) &\quad \geq \quad
\frac{\varepsilon_1}{\varepsilon_2^2} \sum_{\la,\lt}p(\lt,\la)p(a_B,a_C|\lt,\la)p(b_A,b_C|\lt,\la)\\
&\quad = \quad
\zeta_1 \sum_{\lt,\la}p(\lt,\la)p(a_B,a_C|\lt,\la)p(b_A,b_C|\lt,\la), \label{quarantine1}
\end{align}
with $\zeta_1 = \frac{\varepsilon_1}{\varepsilon_2^2}>0$.

%

Let us now further develop the expression \eqref{quarantine1} 
to obtain
\begin{align}\label{selfisolate1}
    p(a_B,a_C,b_A,b_C) &\quad \geq \quad
\zeta_1 \sum_{\lt,\la}p(\lt,\la)p(a_C|\lt,\la)p(b_C|\lt,\la)p(a_B|a_C,\lt,\la)p(b_A|b_C,\lt,\la).
\end{align}

Let us now prove that whenever $p(a_C)>0$ and $p(b_C)>0$,
\begin{align}\label{thetaa}
    p(a_C|\lambda,\gamma) \geq \theta_a\\ \label{thetab}
    p(b_C|\lambda,\gamma) \geq \theta_b,
\end{align}
where $\theta_a,\theta_b > 0$. Let us define $\Omega_{c_A}$ to be the set of all pairs $(\alpha,\beta)$ which "induce" $C_A = c_A$, and similarly $\Omega_{c_B}$ the set of all pairs $(\alpha,\beta)$ which "induce" $C_B = c_B$. The first and last lines in \eqref{independenceCondition2} imply:
\begin{equation}\label{1300}
    \frac{\varepsilon_1}{\varepsilon_2}p(\alpha,\beta) \qleq p(\alpha,\beta|\lambda) \qleq \frac{\varepsilon_2}{\varepsilon_1}p(\alpha,\beta).
\end{equation}
The causal structure of the network implies $p(\alpha,\beta|\lambda) = p(\alpha,\beta|\lambda,\gamma)$. By using this fact and summing inequalities \eqref{1300} over either $(\alpha,\beta)\in \Omega_{c_A}$ or $(\alpha,\beta)\in \Omega_{c_B}$ we obtain the following inequalities:
\begin{align}
    \frac{\varepsilon_1}{\varepsilon_2}p(a_C) &\qleq p(a_C|\lambda,\gamma)\\
    \frac{\varepsilon_1}{\varepsilon_2}p(b_C) &\qleq p(b_C|\lambda,\gamma)
\end{align}
By defining $\theta_a = \min_{a_C}\frac{\varepsilon_1}{\varepsilon_2}p(a_C)$ and $\theta_b = \min_{b_C}\frac{\varepsilon_1}{\varepsilon_2}p(b_C)$ we obtain the inequalities \eqref{thetaa} and \eqref{thetab}.

Hence, by combining eqs. \eqref{selfisolate1}, \eqref{thetaa} and \eqref{thetab} we obtain
\begin{align}\label{selfisolate1bis}
    p(a_B,a_C,b_A,b_C) &\quad \geq \quad
\xi_1 \sum_{\la}p(\la)p(a_B|a_C,\la)p(b_A|b_C,\la)
\end{align}
whenever $p(a_C),p(b_C)>0$. Here, $\xi_1 = \frac{\varepsilon_1^3}{\varepsilon_2^6}\min_{a_C:p(a_C)>0}p(a_C)\min_{b_C:p(b_C)>0}b_C$. Analogously we obtain
\begin{align}\label{selfisolate2bis}
p(a_B,a_C,b_A,b_C) &\quad \leq \quad
\xi_2 \sum_{\la}p(\la)p(a_B|a_C,\la)p(b_A|b_C,\la)
\end{align}
where $\xi_2 = \frac{\varepsilon_2^3}{\varepsilon_1^6}\max_{a_C}p(a_C)\max_{b_C}B_C$, which concludes the proof of Lemma \ref{lem:inequalities}.

\section{Proof for RBBBGB distribution}\label{rbbbgbproof}

In this Appendix we provide the proof that nonlocality of the RBBBGB distribution is robust against small overlap between different sources. The first part of the proof is similar to the one provided in~\cite{Renou_2019} when the sources are independent, while the second part departs significantly. We keep the same notation as in \cite{Renou_2019}, namely: 
\begin{align}
 \ket{\uparrow} = \ket{01}  , &\qquad
 \ket{\chi_0} = u_0\ket{00} + v_0\ket{11},\\  \ket{\downarrow} = \ket{10}, &\qquad \ket{\chi_1} = u_1\ket{00} + v_1\ket{11},
\end{align}
where $u_0 = u = -v_1$ and $v_0 = \sqrt{1-u^2} = u_1 = v$. The RBBBGB distribution has many useful properties. The number of parties outputing either $\chi_0$ or $\chi_1$ must be odd. When it comes to other outputs, the following equations are satisfied:
\begin{align}
    p(a=\uparrow,b=\uparrow) = p(a=\downarrow,b=\downarrow) &= 0 \label{eq:zeroProb}\\
    p(a=\uparrow,b=\downarrow) = p(b=\uparrow,c=\downarrow) = p(c=\uparrow,a=\downarrow) &= s_0^4s_1^2 \label{eq:updown}\\
    p(c=\downarrow,a=\uparrow) &= s_0^2s_1^4\\ \label{ciupdown}
    p(\chi_i,\uparrow,\downarrow) &=  s_0^4s_1^2u_i^2\\ \label{cidownup}
    p(\chi_i,\downarrow,\uparrow) &=  s_0^4s_1^2v_i^2\\ \label{cicjck}
    p(\chi_i,\chi_j,\chi_k) &= (s_0^3u_iu_ju_k + s_1^3v_iv_jv_k)^2,
\end{align}
and equivalent expressions for cyclic permutations of the parties.

\begin{lem}\label{lambdaSplitting} Consider the coarse graining of the outputs $\{\uparrow,\downarrow,\chi = \{\chi_0,\chi_1\}\}$. A $\varepsilon_1$-trilocal model with $\varepsilon_1>0$ for the resulting statistics has the following form: the hidden variables are partitioned in the following way:
$\Lambda_A = \Lambda_{A,0}\sqcup \Lambda_{A,1}$, $\Lambda_B = \Lambda_{B,0}\sqcup \Lambda_{B,1}$ and $\Lambda_C = \Lambda_{C,0}\sqcup \Lambda_{C,1}$, and the outcomes are determined by the following rules:
\begin{itemize}
    \item Alice outputs $\uparrow$ if she receives $\beta \in \Lambda_{B,0}$ and $\gamma \in \Lambda_{C,1}$;
    \item Alice outputs $\downarrow$ if she receives $\beta \in \Lambda_{B,1}$ and $\gamma \in \Lambda_{C,0}$;
    \item she outputs $\chi$ otherwise,
\end{itemize}
and analagously for Bob and Charlie. The probabilities to obtain hidden variable values from different partitions are:
\begin{align}\label{singlemarg}
    p(\Lambda_{A,i}) = p(\Lambda_{B,i}) = p(\Lambda_{C,i}) &= s_i^2
\end{align}
\end{lem}
\begin{proof}
The proof follows closely the one provided for fully independent sources in~\cite{Renou_2019}. Define the sets
\begin{align*}
    \Lambda_{A,0}=\{\alpha|\exists \gamma : b(\alpha,\gamma)=\downarrow \}, &\quad \Lambda_{A,0}'=\{\alpha|\exists \gamma : b(\alpha,\gamma)=\uparrow \} \\
    \Lambda_{A,1}'=\{\alpha|\exists \beta : c(\alpha,\beta)=\uparrow \}, &\quad \Lambda_{A,1}=\{\alpha|\exists \beta : c(\alpha,\beta)=\downarrow \} \\
    \Lambda_{B,0}=\{\beta|\exists \gamma : a(\beta,\gamma)=\downarrow \}, &\quad \Lambda_{B,0}'=\{\beta|\exists \gamma : a(\beta,\gamma)=\uparrow \} \\
    \Lambda_{B,1}'=\{\beta|\exists \alpha : c(\alpha,\beta)=\uparrow \}, &\quad \Lambda_{B,1}=\{\beta|\exists \alpha : c(\alpha,\beta)=\downarrow \} \\
    \Lambda_{C,0}=\{\gamma|\exists \alpha : a(\alpha,\gamma)=\downarrow \}, &\quad \Lambda{C_0}'=\{\gamma|\exists \alpha : a(\alpha,\gamma)=\uparrow \} \\
    \Lambda_{C,1}'=\{\gamma|\exists \beta : b(\beta,\beta)=\uparrow \}, &\quad \Lambda_{C,1}=\{\gamma|\exists \beta : b(\beta,\beta)=\downarrow \}.
\end{align*}
Because of condition~\eqref{eq:zeroProb}, we must have that:
\begin{align}
     \Lambda_{A,0}\cap \Lambda_{A,1} = \Lambda_{A,0}'\cap \Lambda_{A,1}' = \Lambda_{B,0}\cap \Lambda_{B,1} = \Lambda_{B,0}'\cap \Lambda_{B,1}' = \Lambda_{C,0}\cap \Lambda_{C,1} = \Lambda_{C,0}'\cap \Lambda_{C,1}' = \varnothing.
\end{align}
Indeed, if some value of $\alpha$ could belong to both $\Lambda_{A,0}$ and $\Lambda_{A,1}$, it would allow simultaneously for Bob to produce $b(\alpha,\gamma)=\downarrow$ when he also receives $\gamma=\gamma^*$ from the $\gamma$ source, and for Charlies to produce $c(\alpha,\beta)=\downarrow$ for some $\beta=\beta^*$. Since $p(b=\downarrow,c=\downarrow)=0$, these two events should never occur at the same time. However, the measurement dependence condition~\eqref{independenceCondition} guarantees that such event will occur, because the joint probability $p(\beta^*,\gamma^*|\lambda)\geq \epsilon_1 p(\beta^*)p(\gamma^*)>0$ is strictly positive. Hence $\Lambda_{A,0}$ and $\Lambda_{A,1}$ must be disjoint.

Moreover, we also have that $\Lambda_A=\Lambda_{A,0}\sqcup \Lambda_{A,1}$, and similarly for $\Lambda_B$ and $\Lambda_C$. Indeed, if there was an $\alpha^*\not\in \Lambda_{A,0}\sqcup \Lambda_{A,1}$, then for all $\beta$, $\gamma$, $b(\alpha^*,\gamma)\neq\uparrow$ and $c(\alpha^*,\beta)\neq\uparrow$. Given Eq.~\eqref{eq:zeroProb} and the fact that the number of $\chi$ output must be odd, we can deduce then that we must have $a(\beta,\gamma)\neq\downarrow$ for all $\beta$, $\gamma$, which is incompatible with Eq.~\eqref{eq:updown}.
\end{proof}

Let us, now, following the approach of \cite{Renou_2019} introduce the following variable
\begin{align}\nonumber
    q(i,j,k,t) &= p(a= \chi_i,b = \chi_j,c= \chi_k,(\alpha,\beta,\gamma) \in \Lambda_{A,t} \times \Lambda_{B,t} \times \Lambda_{C,t}|(\alpha,\beta,\gamma) \in \Lambda_{A,0} \times \Lambda_{B,0} \times \Lambda_{C,0} \cup \Lambda_{A,1} \times \Lambda_{B,1} \times \Lambda_{C,1})\\
    &= \frac{1}{s_0^6+s_1^6}p(a= \chi_i,b = \chi_j,c= \chi_k,(\alpha,\beta,\gamma) \in \Lambda_{A,t} \times \Lambda_{B,t} \times \Lambda_{C,t}).
\end{align}
$q(i,j,k,t)$ is a valid probability distribution and the aim is to find whether its marginals are consistent with the existence of local hidden variable models (which have to satisfy lemma~\ref{lambdaSplitting}). For that we need several inter-steps. At this stage, we introduced the short notation $p(\Lambda_{A,i}) \equiv p(\alpha \in \Lambda_{A,i})$. We will use a similar notation for other hidden variables and joint probabilities. Moreover, defining the set $\Omega^{00}_{\chi_i}=\{(\beta,\gamma)\in\Lambda_{B,0}\times\Lambda_{C,0}:a=\chi_i\}$, inequalities \eqref{condition} and \eqref{condition2} imply:
\begin{align}\label{ge}
    \varepsilon_1 p(\Lambda_{A,t})\sum_{(\beta,\gamma)\in\Omega^{00}_{\chi_i}} p(\beta)p(\gamma)\quad \leq\quad &p(a=\chi_i,\Lambda_{A,t},\Lambda_{B,0},\Lambda_{C,0})\quad \leq\quad \varepsilon_2 p(\Lambda_{A,t})\sum_{(\beta,\gamma)\in\Omega^{00}_{\chi_i}} p(\beta)p(\gamma)\\ \label{ne}
    \frac{\varepsilon_1}{\varepsilon_2}p(\Lambda_{A,t})p(a=\chi_i,\Lambda_{B,0},\Lambda_{C,0})\quad \leq\quad &p(a=\chi_i,\Lambda_{A,t},\Lambda_{B,0},\Lambda_{C,0})\quad \leq\quad \frac{\varepsilon_2}{\varepsilon_1}p(\Lambda_{A,t})p(a=\chi_i,\Lambda_{B,0},\Lambda_{C,0})\\ \label{ve}
    \frac{\varepsilon_1^2}{\varepsilon_2^2}\frac{p(\Lambda_{A,0})}{p(\Lambda_{A,1})}p(a=\chi_i,\Lambda_{A,1},\Lambda_{B,0},\Lambda_{C,0})\quad \leq\quad &p(a=\chi_i,\Lambda_{A,0},\Lambda_{B,0},\Lambda_{C,0})\quad \leq\quad \frac{\varepsilon_2^2}{\varepsilon_1^2}\frac{p(\Lambda_{A,0})}{p(\Lambda_{A,1})}p(a=\chi_i,\Lambda_{A,1}\Lambda_{B,0},\Lambda_{C,0}).
\end{align}
The first pair of inequalities is obtained by multiplying ineqs. \eqref{condition} and \eqref{condition2} by $p(\lambda)$, and then successively summing over all values of $\lambda$, over all values $\alpha \in \Lambda_{A,t}$ with $t=0,1$, and over all values of $\beta \in \Lambda_{B,0}, \gamma \in \Lambda_{C,0}$ that "induce" the result $a = \chi_i$. $\varepsilon_2$ is the maximum over $\varepsilon(\alpha,\beta,\gamma)$ and over $t$ for all $\alpha,\beta,\gamma$ which are involved in the summation. The second pair of inequalities is obtained from \eqref{condition}, \eqref{condition2} by multiplying with $p(\lambda)$ and summing over $\alpha$ and $\lambda$ and over the values of $\beta$ and $\gamma$ which belong to $\Lambda_{B,0}$ and $\Lambda_{C,0}$ respectively, and "induce" $a = \chi_i$. This leads to 
\begin{equation}
\varepsilon_1 \sum_{(\beta,\gamma)\in\Omega^{00}_{\chi_i}} p(\beta)p(\gamma) \quad\leq\quad p(a=\chi_i,\Lambda_{B,0},\Lambda_{C,0}) \quad\leq\quad \varepsilon_2 \sum_{(\beta,\gamma)\in\Omega^{00}_{\chi_i}} p(\beta)p(\gamma)
\end{equation}
and therefore
\begin{equation}
\varepsilon_1 p(\Lambda_{A,t}) \sum_{(\beta,\gamma)\in\Omega^{00}_{\chi_i}} p(\beta)p(\gamma) \quad\leq\quad p(\Lambda_{A,t}) p(a=\chi_i,\Lambda_{B,0},\Lambda_{C,0}) \quad\leq\quad \varepsilon_2 p(\Lambda_{A,t}) \sum_{(\beta,\gamma)\in\Omega^{00}_{\chi_i}} p(\beta)p(\gamma)
\end{equation}
after multiplication by $p(\Lambda_{A,t})$. Using \eqref{ge} then yields \eqref{ne}. A similar trick is used to obtain \eqref{ve}, namely starting from \eqref{ne} with $t=1$, i.e.
\begin{equation}
\frac{\varepsilon_1}{\varepsilon_2}p(a=\chi_i,\Lambda_{A,1},\Lambda_{B,0},\Lambda_{C,0})    \quad\leq\quad p(\Lambda_{A,1})p(a=\chi_i,\Lambda_{B,0},\Lambda_{C,0})\quad \leq \quad \frac{\varepsilon_2}{\varepsilon_1}p(a=\chi_i,\Lambda_{A,1},\Lambda_{B,0},\Lambda_{C,0}),
\end{equation}
and combining it with \eqref{ne} with $t=0$.

Let us now use inequalities \eqref{ge}-\eqref{ve} to calculate bounds on the marginal probability distributions of $q(i,j,k,t)$. We start with $q(i,j,k)$:
\begin{align}
    q(i,j,k) &= q(i,j,k,0) + q(i,j,k,1)\\
    &= \frac{p(a=\chi_i,b=\chi_j,c=\chi_k,t=0)}{p((\alpha,\beta,\gamma)\in \cup_{t=0}^1\Lambda_{A,t}\times\Lambda_{B,t}\times\Lambda_{C,t})} + \frac{p(a=\chi_i,b=\chi_j,c=\chi_k,t=1)}{p((\alpha,\beta,\gamma)\in \cup_{t=0}^1\Lambda_{A,t}\times\Lambda_{B,t}\times\Lambda_{C,t})} \\
    &= \frac{p(a=\chi_i,b=\chi_j,c=\chi_k)}{p(\Lambda_{A,0},\Lambda_{B,0},\Lambda_{C,0}) + p(\Lambda_{A,1},\Lambda_{B,1},\Lambda_{C,1})},
\end{align}
where we used that $p(a=\chi_i,b=\chi_j,c=\chi_k,\alpha\in\Lambda_{A,t_1},\beta\in\Lambda_{B,t_2},\gamma\in\Lambda_{C,t_3}) = 0$ if $(t_1, t_2, t_3)\not\in\{(0,0,0),(1,1,1)\}$. The denominator terms satisfy the following bounds
\begin{equation}
    \varepsilon_1p(\Lambda_{A,t})p(\Lambda_{B,t})p(\Lambda_{C,t})\quad \leq \quad p(\Lambda_{A,t},\Lambda_{B,t},\Lambda_{C,t}) \quad \leq\quad \varepsilon_2p(\Lambda_{A,t})p(\Lambda_{B,t})p(\Lambda_{C,t})\quad \forall t,
\end{equation}
which implies 
\begin{equation}\label{margijk}
\frac{1}{\varepsilon_2}\frac{(s_0^3u_iu_ju_k + s_1^3v_iv_jv_k)^2}{s_0^6 + s_1^6} \quad\leq\quad q(i,j,k)\quad\leq\quad \frac{1}{\varepsilon_1}\frac{(s_0^3u_iu_ju_k + s_1^3v_iv_jv_k)^2}{s_0^6 + s_1^6},
\end{equation}
where we used \eqref{cicjck} and \eqref{singlemarg}.

Similarly, the following marginal directly has bounds:
\begin{align}
 \frac{p(a=\chi_i,(\alpha,\beta,\gamma)\in \Lambda_{A,0}\times\Lambda_{B,0}\times\Lambda_{C,0})}{\varepsilon_2(s_0^6+s_1^6)} \qleq q(i,t=0)\qleq \frac{p(a=\chi_i,(\alpha,\beta,\gamma)\in \Lambda_{A,0}\times\Lambda_{B,0}\times\Lambda_{C,0})}{\varepsilon_1(s_0^6+s_1^6)}.
\end{align}
By using the inequalities \eqref{ve} we obtain
\begin{align}
    \frac{\varepsilon_1^2}{\varepsilon_2^3}\frac{p(\Lambda_{A,0})}{p(\Lambda_{A,1})}\frac{p(a=\chi_i,\Lambda_{A,1},\Lambda_{B,0},\Lambda_{C,0})}{s_0^6+s_1^6}\qleq &q(i,t=0)\qleq \frac{\varepsilon_2^2}{\varepsilon_1^3}\frac{p(\Lambda_{A,0})}{p(\Lambda_{A,1})}\frac{p(a=\chi_i,\Lambda_{A,1},\Lambda_{B,0},\Lambda_{C,0})}{s_0^6+s_1^6}\\ \label{sw}
      \frac{\varepsilon_1^2}{\varepsilon_2^3}\frac{s_0^6u_i^2}{s_0^6+s_1^6}\qleq &q(i,t=0)\qleq \frac{\varepsilon_2^2}{\varepsilon_1^3}\frac{s_0^6u_i^2}{s_0^6+s_1^6},
\end{align}
where for the second set of inequalities we used \eqref{ciupdown} and \eqref{singlemarg}. In a similar manner we obtain the following bounds on the single-party marginals:
\begin{align}\label{it}
    \frac{\varepsilon_1^2}{\varepsilon_2^3}\frac{s_1^6v_i^2}{s_0^6+s_1^6}\qleq &q(i,t=1)\qleq \frac{\varepsilon_2^2}{\varepsilon_1^3}\frac{s_1^6v_i^2}{s_0^6+s_1^6}\\ \label{ze}
    \frac{\varepsilon_1^2}{\varepsilon_2^3}\frac{s_0^6u_j^2}{s_0^6+s_1^6}\qleq &q(j,t=0)\qleq \frac{\varepsilon_2^2}{\varepsilon_1^3}\frac{s_0^6u_j^2}{s_0^6+s_1^6}\\ \label{rl}
        \frac{\varepsilon_1^2}{\varepsilon_2^3}\frac{s_1^6v_j^2}{s_0^6+s_1^6}\qleq &q(j,t=1)\qleq \frac{\varepsilon_2^2}{\varepsilon_1^3}\frac{s_1^6v_j^2}{s_0^6+s_1^6}\\ \label{an} \frac{\varepsilon_1^2}{\varepsilon_2^3}\frac{s_0^6u_k^2}{s_0^6+s_1^6}\qleq &q(k,t=0)\qleq \frac{\varepsilon_2^2}{\varepsilon_1^3}\frac{s_0^6u_k^2}{s_0^6+s_1^6}\\ \label{d}
        \frac{\varepsilon_1^2}{\varepsilon_2^3}\frac{s_1^6v_k^2}{s_0^6+s_1^6}\qleq &q(k,t=1)\qleq \frac{\varepsilon_2^2}{\varepsilon_1^3}\frac{s_1^6v_k^2}{s_0^6+s_1^6}
\end{align}

Let us now introduce the symmetrized distribution $\tilde{q}(i,j,k,t)$ defined as
\begin{align}
    \tilde{q}(i,j,k,t) = \frac{1}{6}(q(i,j,k,t) + q(i,k,j,t) + q(j,i,k,t) + q(j,k,i,t) + q(k,i,j,t)+ q(k,j,i,t)).
\end{align}
Its marginals clearly satisfy the same bounds as those given for the marginals of $q(i,j,k,t)$ in ineqs. \eqref{sw}-\eqref{d}, for instance.
\begin{align}\label{zetilde}
    \frac{\varepsilon_1^2}{\varepsilon_2^3}\frac{s_0^6u_k^2}{s_0^6+s_1^6}\qleq &\tilde q(k,t=0)\qleq \frac{\varepsilon_2^2}{\varepsilon_1^3}\frac{s_0^6u_k^2}{s_0^6+s_1^6}\\ \label{ittilde}
    \frac{\varepsilon_1^2}{\varepsilon_2^3}\frac{s_1^6v_k^2}{s_0^6+s_1^6}\qleq &\tilde q(k,t=1)\qleq \frac{\varepsilon_2^2}{\varepsilon_1^3}\frac{s_1^6v_k^2}{s_0^6+s_1^6}
\end{align}
Furthermore, consider $\xi_{ijk} \equiv \tilde{q}(i,j,k,0)-\tilde{q}(i,j,k,1)$. The following relations hold:
\begin{align}\label{qtildexi0}
\tilde{q}(i,j,k,0) & = (\tilde{q}(i,j,k)+\xi_{ijk})/2\\ \label{qtildexi1}
\tilde{q}(i,j,k,1) & = (\tilde{q}(i,j,k)-\xi_{ijk})/2.
\end{align}
The symmetry allows to simplify the notation, as there are only four different variables $\xi_{ijk}$: $\xi_{0} \equiv \xi_{000}, \xi_{1} \equiv \xi_{001} = \xi_{010} = \xi_{100}, \xi_2 \equiv \xi_{110} = \xi_{101} = \xi_{011}, \xi_3 \equiv \xi_{111}$. Hence, given \eqref{margijk} we can obtain the following bounds:
\begin{align}
    \frac{1}{2}\sum_{i,j}\frac{1}{\varepsilon_2}\frac{(s_0^3u_iu_ju_k + s_1^3v_iv_jv_k)^2}{s_0^6 + s_1^6} + \xi_{ijk} \qleq &\tilde{q}(k,t=0) \qleq \frac{1}{2}\sum_{i,j}\frac{1}{\varepsilon_1}\frac{(s_0^3u_iu_ju_k + s_1^3v_iv_jv_k)^2}{s_0^6 + s_1^6} + \xi_{ijk}\\
    &\Updownarrow& \nonumber\\ \label{B35}
     \frac{1}{2}\left(\frac{1}{\varepsilon_2}\frac{s_0^6u_k^2 + s_1^6v_k^2}{s_0^6 + s_1^6} + \sum_{i,j}\xi_{ijk}\right) \qleq &\tilde{q}(k,t=0) \qleq \frac{1}{2}\left(\frac{1}{\varepsilon_1}\frac{s_0^6u_k^2 + s_1^6v_k^2}{s_0^6 + s_1^6} + \sum_{i,j}\xi_{ijk}\right)\\
     \frac{1}{2}\left(\frac{1}{\varepsilon_2}\frac{s_0^6u_k^2 + s_1^6v_k^2}{s_0^6 + s_1^6} - \sum_{i,j}\xi_{ijk}\right) \qleq &\tilde{q}(k,t=1) \qleq \frac{1}{2}\left(\frac{1}{\varepsilon_1}\frac{s_0^6u_k^2 + s_1^6v_k^2}{s_0^6 + s_1^6} - \sum_{i,j}\xi_{ijk}\right)
\end{align}
Note now that $\sum_{ij}\xi_{ij0} = \xi_0 + 2\xi_1 + \xi_2$ and $\sum_{ij}\xi_{ij1} = \xi_1 + 2\xi_2 + \xi_3$. Hence, using \eqref{B35} and \eqref{zetilde}, the following bound hold for $\xi_0$:
\begin{align}
    \frac{\varepsilon_1^2}{\varepsilon_2^3}\frac{2s_0^6u^2}{s_0^6+s_1^6} - \frac{1}{\varepsilon_1}\frac{(s_0^6u^2 + s_1^6v^2)}{s_0^6 + s_1^6} -2\xi_1 -\xi_2 \qleq &\xi_0 \qleq \frac{\varepsilon_2^2}{\varepsilon_1^3}\frac{2s_0^6u^2}{s_0^6+s_1^6} - \frac{1}{\varepsilon_2}\frac{(s_0^6u^2 + s_1^6v^2)}{s_0^6 + s_1^6} -2\xi_1 -\xi_2\\
    &\Updownarrow& \nonumber\\ \label{ksi0}
    \frac{1}{\varepsilon_1\varepsilon_2^3}\frac{s_0^6u^2(2\varepsilon_1^3-\varepsilon_2^3) - \varepsilon_2^3s_1^6v^2}{s_0^6 + s_1^6} -2\xi_1 -\xi_2 \qleq &\xi_0 \qleq \frac{1}{\varepsilon_2\varepsilon_1^3}\frac{s_0^6u^2(2\varepsilon_2^3-\varepsilon_1^3) - \varepsilon_1^3s_1^6v^2}{s_0^6 + s_1^6} -2\xi_1 -\xi_2.
\end{align}
Similarly, the following bound holds for $\xi_3$:
%
%
\begin{align}\label{ksi3}
       \frac{1}{\varepsilon_2\varepsilon_1^3}\frac{s_1^6v^2(\varepsilon_1^3-2\varepsilon_2^3) - \varepsilon_1^3s_0^6u^2}{s_0^6 + s_1^6} -\xi_1 -2\xi_2 \qleq &\xi_3 \qleq \frac{1}{\varepsilon_1\varepsilon_2^3}\frac{s_1^6v^2(\varepsilon_2^3-2\varepsilon_1^3) + \varepsilon_2^3s_0^6u^2}{s_0^6 + s_1^6} -\xi_1 -2\xi_2.
\end{align}
Now considering \eqref{qtildexi0}, \eqref{qtildexi1} and \eqref{margijk} we can bound $\tilde{q}(0,0,0,1)$ and $\tilde{q}(1,1,1,0)$:
%
\begin{align}\label{zzz}
    \frac{1}{\varepsilon_2}\frac{(s_0^3u^3 + s_1^3v^3)^2}{2(s_0^6 + s_1^6)} - \frac{\xi_0}{2} \qleq &\tilde{q}(0,0,0,1) \qleq \frac{1}{\varepsilon_1}\frac{(s_0^3u^3 + s_1^3v^3)^2}{2(s_0^6 + s_1^6)} - \frac{\xi_0}{2}\\ \label{ooo}
    \frac{1}{\varepsilon_2}\frac{(s_0^3v^3 - s_1^3u^3)^2}{2(s_0^6 + s_1^6)} + \frac{\xi_3}{2} \qleq &\tilde{q}(1,1,1,0) \qleq \frac{1}{\varepsilon_1}\frac{(s_0^3v^3 - s_1^3u^3)^2}{2(s_0^6 + s_1^6)} + \frac{\xi_3}{2}
\end{align}
These inequalities allow to put further bounds on $\xi_2$. First, taking into account \eqref{ksi0}, \eqref{zzz} and $\tilde q(0,0,0,1)\geq 0$ we obtain:
%
\begin{align}
    \xi_2 &\qgeq \frac{1}{\varepsilon_1\varepsilon_2^3}\frac{s_0^6u^2(2\varepsilon_1^3-\varepsilon_2^3) - \varepsilon_2^3s_1^6v^2}{s_0^6 + s_1^6} -2\xi_1 -  \frac{1}{\varepsilon_1}\frac{(s_0^3u^3 + s_1^3v^3)^2}{s_0^6 + s_1^6}\\
    &\quad = \quad \frac{1}{\varepsilon_1\varepsilon_2^3}\frac{s_0^6u^2(2\varepsilon_1^3-\varepsilon_2^3) - \varepsilon_2^3s_1^6v^2 - \varepsilon_2^3(s_0^3u^3 + s_1^3v^3)^2}{s_0^6 + s_1^6} -2\xi_1 
\end{align}
Similarly, inequalities \eqref{ksi3}, \eqref{ooo} and $\tilde q(1,1,1,0)\geq 0$ imply:
%
\begin{align}
    \xi_2 &\qleq \frac{1}{\varepsilon_1\varepsilon_2^3}\frac{s_1^6v^2(\varepsilon_2^3-2\varepsilon_1^3) + \varepsilon_2^3s_0^6u^2}{2(s_0^6 + s_1^6)} -\frac{\xi_1}{2} +  \frac{1}{\varepsilon_1}\frac{(s_0^3v^3 - s_1^3u^3)^2}{2(s_0^6 + s_1^6)}\\
    &= \frac{1}{\varepsilon_1\varepsilon_2^3}\frac{s_1^6v^2(\varepsilon_2^3-2\varepsilon_1^3) + \varepsilon_2^3s_0^6u^2 + \varepsilon_2^3(s_0^3v^3 - s_1^3u^3)^2}{2(s_0^6 + s_1^6)} -\frac{\xi_1}{2} 
\end{align}
The last two inequalities allow to find a bound for $\xi_1$:
%
\begin{align}
    \frac{3}{2}\xi_1 &\qgeq \frac{1}{\varepsilon_1\varepsilon_2^3}\frac{s_0^6u^2(2\varepsilon_1^3-\varepsilon_2^3) - \varepsilon_2^3s_1^6v^2 - \varepsilon_2^3(s_0^3u^3 + s_1^3v^3)^2}{s_0^6 + s_1^6} - \frac{1}{\varepsilon_1\varepsilon_2^3}\frac{s_1^6v^2(\varepsilon_2^3-2\varepsilon_1^3) + \varepsilon_2^3s_0^6u^2 + \varepsilon_2^3(s_0^3v^3 - s_1^3u^3)^2}{2(s_0^6 + s_1^6)}\\ \label{unige}
    &\quad =\quad \frac{1}{\varepsilon_1\varepsilon_2^3}\frac{s_0^6u^2(4\varepsilon_1^3-3\varepsilon_2^3) + s_1^6v^2(2\varepsilon_1^3-3\varepsilon_2^3) - 2\varepsilon_2^3(s_0^3u^3 + s_1^3v^3)^2 - \varepsilon_2^3(s_0^3v^3 - s_1^3u^3)^2}{2(s_0^6 + s_1^6)}.
\end{align}
Finally, we can get rid of $\xi_1$, by utilizing the positivity of $\tilde{q}(0,0,1,1)$ with \eqref{qtildexi1} and \eqref{margijk} 
\begin{align}\label{gap}
    0 \qleq &\tilde{q}(0,0,1,1) \qleq \frac{1}{\varepsilon_1}\frac{(s_0^6u^2v - s_1^6v^2u)^2}{2(s_0^6 + s_1^6)} - \frac{\xi_1}{2}
\end{align}
So for those values of $\varepsilon_1$ and $\varepsilon_2$ for which \eqref{unige} and \eqref{gap} have no common solutions the RBBBGB distribution remains nonlocal. Thus, the final inequality, giving the region of allowed $\varepsilon_1$ and $\varepsilon_2$ is: 
\begin{multline}
    \frac{1}{\varepsilon_1\varepsilon_2^3}\frac{s_0^6u^2(4\varepsilon_1^3-3\varepsilon_2^3) + s_1^6v^2(2\varepsilon_1^3-3\varepsilon_2^3) - 2\varepsilon_2^3(s_0^3u^3 + s_1^3v^3)^2 - \varepsilon_2^3(s_0^3v^3 - s_1^3u^3)^2}{2(s_0^6 + s_1^6)} \qgeq \frac{3}{\varepsilon_1}\frac{(s_0^6u^2v - s_1^6v^2u)^2}{2(s_0^6 + s_1^6)}
\end{multline}
which is equivalent to
\begin{multline}
    s_0^6u^2(4\varepsilon_1^3-3\varepsilon_2^3) + s_1^6v^2(2\varepsilon_1^3-3\varepsilon_2^3) - 2\varepsilon_2^3(s_0^3u^3 + s_1^3v^3)^2 - \varepsilon_2^3(s_0^3v^3 - s_1^3u^3)^2\qgeq 3\varepsilon_2^3(s_0^6u^2v - s_1^6v^2u)^2.
\end{multline}

\section{Nonlocality in square networks}\label{app:square}

\begin{figure*}
\begin{subfigure}{.49\textwidth}
  \centering
  \includegraphics[width=.49\linewidth]{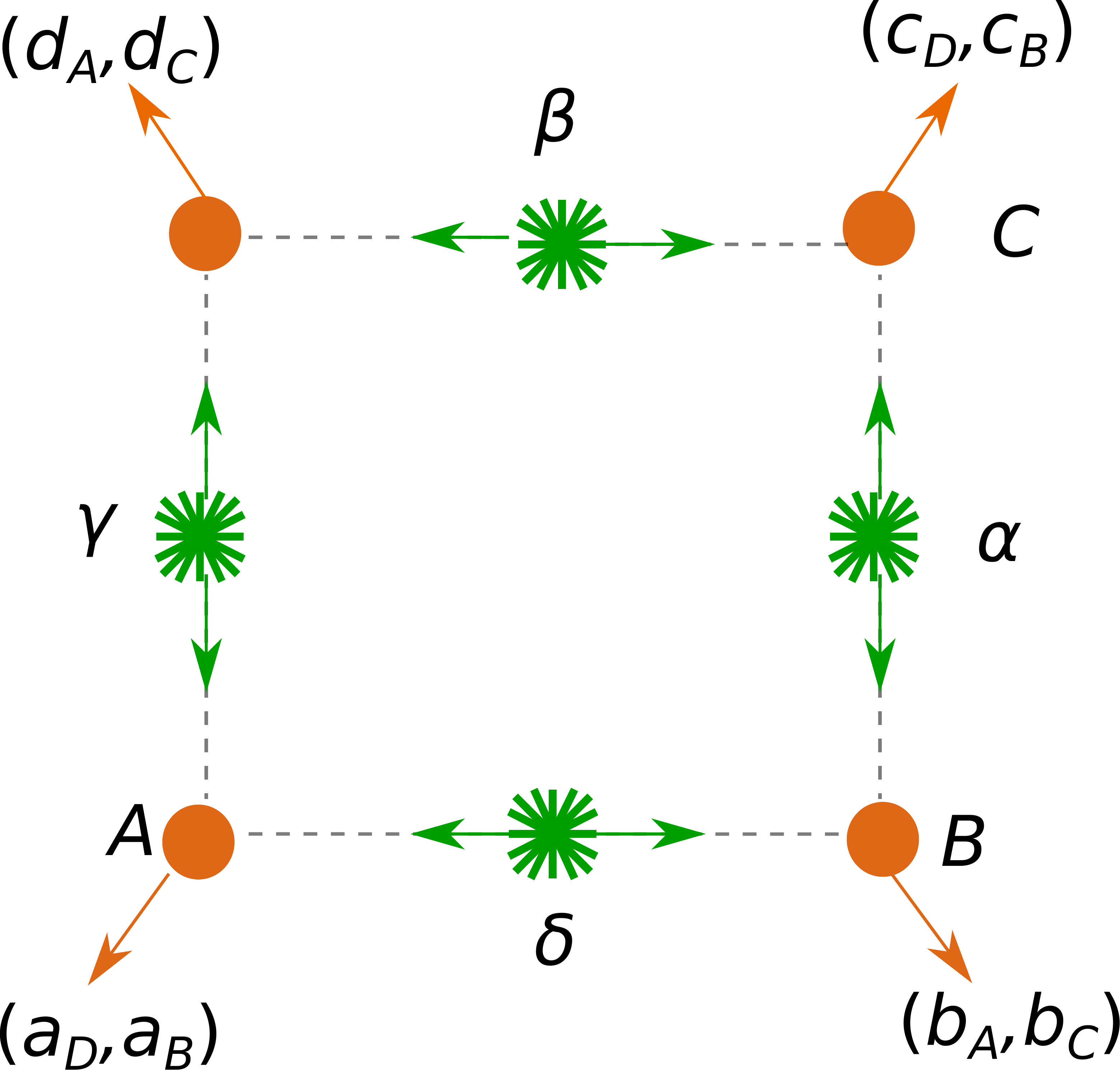}  
  \caption{Square network with independent sources}
  \label{fig:IndependentSquare}
\end{subfigure}
\begin{subfigure}{.49\textwidth}
  \centering
  \includegraphics[width=.49\linewidth]{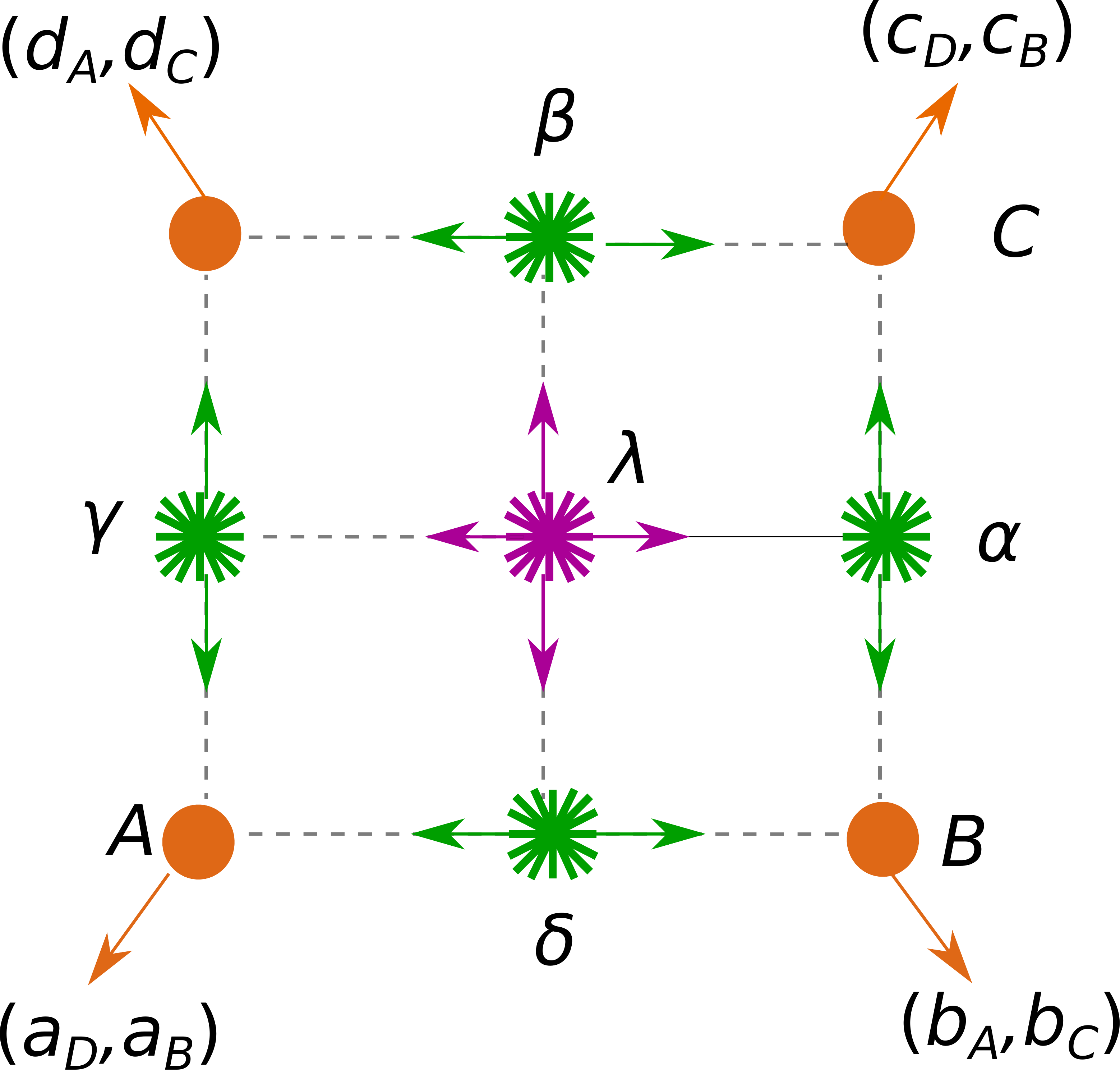}  
  \caption{Square network with partially independent sources}
  \label{fig:OverlappingSquare}
\end{subfigure}
\caption{Two types of square networks}
\label{fig:squarenetworks}
\end{figure*}

The framework we introduced for examining nonlocality in networks for partially independent sources can be applied to networks different than the triangular one. In this section we consider the square network. Let us first consider the square network with fully independent sources (see Fig. \ref{fig:IndependentSquare}). In such network the behaviour satisfying the following conditions
\begin{align}
    p(A_D = D_A \land C_D = D_C) = 1
\end{align}
is incompatible with local hidden variable models whenever the behaviour $\{p(a_B,b,c_B|d_A,d_C)\}$ is non-bilocal. The example of such behaviour is given in section \ref{app:independentSquare}. As the example of triangle-nonlocal behaviour given by Fritz \cite{Fritz_2012} can be seen as embedding standard Bell nonlocal behaviours into a triangle network, so non-bilocal behaviours can be embedded into a square network. 

Consider now a square network with partially independent sources, as shown on Fig. \ref{fig:OverlappingSquare}. The fourth party is named Daisy. The partial independence of the sources is ensured through the following criteria
\begin{equation}\label{sqind1}
  \varepsilon_1p(\alpha)p(\beta)p(\gamma)p(\delta) \leq p(\ld,\lp,\la,\delta|\lambda),
\end{equation}
where  $\varepsilon_1 > 0$ and
\begin{equation}\label{sqind2}
    p(\ld,\lp,\la,\delta|\lambda) \leq \varepsilon_2(\ld,\lp,\la,\delta,\lambda)p(\alpha)p(\beta)p(\gamma)p(\delta),
\end{equation}
where $\varepsilon_2(\ld,\lp,\la,\delta,\lambda) \in [1,1/p(\ld)p(\lp)p(\la)p(\delta))$ for all $\ld,\lp,\la,\delta,\lambda$. For certain values of $\varepsilon_1$ and $\varepsilon_2(\alpha,\beta,\gamma,\delta,\lambda)$ there are quantum realizations which are provable nonlocal. For example, consider the following quantum strategy: Alice and Bob share with Charlie the singlet state $\ket{\psi^-}$, while with Daisy they share classically correlated state $(\proj{00} + \proj{11})/2$. Daisy measures her qubits in the computational basis. Bob performs the Bell state measurement, while Alice's and Charlie's measurements have the following form: $\M_{a_Da_B}^{\rA_\rD\rA_\rB} = \proj{a_D}\otimes\sigma_{a_D}^{a_B}$ and $\M_{c_Bc_D}^{\rC_\rB\rC_\rD} = \sigma_{c_D}^{c_B}\otimes\proj{c_D}$, where $\sigma_n^{m} = (\idd + (-1)^m(\sigma_z + (-1)^n \sigma_x)/\sqrt{2})/2$. In section \ref{app:overlappingSquare} we show that the behaviour obtained by using this quantum strategy is provably nonlocal for certain values of $\varepsilon_1$. The proof is analogous to the proof of Theorem \ref{theorem}: the nonlinear inequality used to detect non-bilocality can be used to construct another nonlinear inequality, depending on $\varepsilon_1$ and $\varepsilon_2 = \max_{\alpha,\beta,\gamma,\delta,\lt}\varepsilon_2(\alpha,\beta,\gamma,\delta,\lt)$, whose violation witnesses nonlocality of the square behaviour.

\subsection{The example of a square-nonlocal behaviour}\label{app:independentSquare}

Certain types of  non-bilocal behaviour can be the used as a basis for creating a square-nonlocal behaviour. The scenario for non-bilocality involves three parties, two lateral, Alice and Charlie, and one central, Bob (see fig. \ref{fig:Bilocality}). Different parties can have various number of inputs, but specifically interesting in the context of square nonlocality without inputs are those bilocality scenarios in which Bob has no inputs. The notation we use here is chosen so that the transition to the square networks is more natural. Namely, in local models, the hidden variable shared between Alice and Bob takes value $\delta$, while the one shared between Bob and Charlie takes value $\alpha$. Alice's binary output is $a_B$, Charlie's $c_B$, while Bob outputs two bits $(b_A,b_C)$. The binary input for Alice is $d_A$ and for Charlie $d_C$.

\begin{figure}[h!]
  \centering
  \includegraphics[width=.49\linewidth]{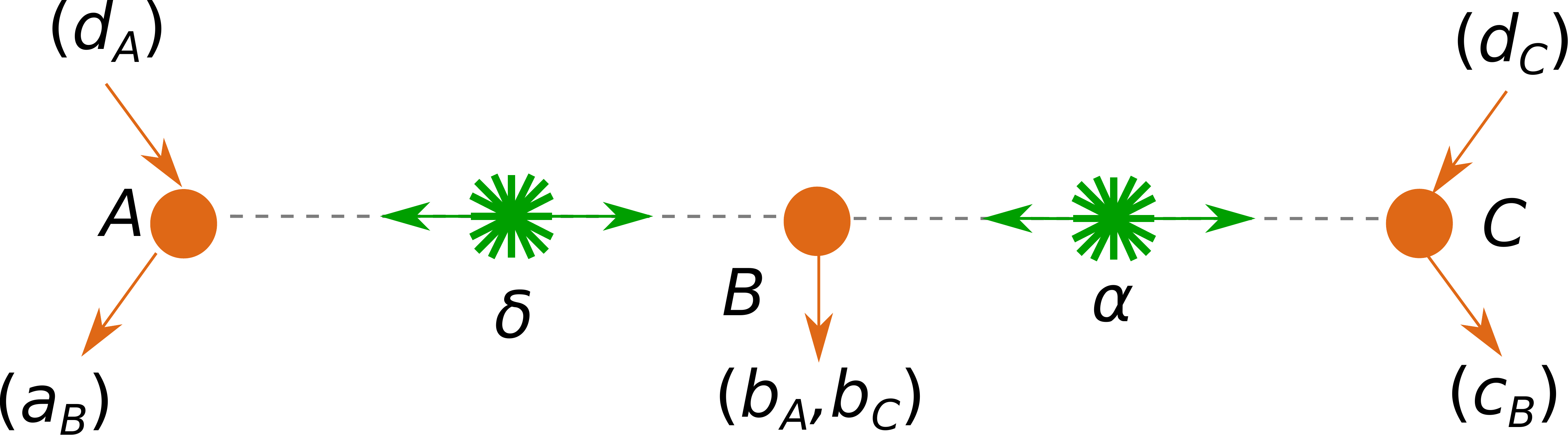}  
  \caption{Bilocality scenario with two independent sources $\delta$ and $\alpha$.}
  \label{fig:Bilocality}
\end{figure}

In this appendix we use non-bilocality inequality introduced in \cite{biloc2} as a building block for square nonlocality. Following \cite{biloc2}, let us introduce the correlations
\begin{equation}
    \langle A_{d_A}B^yC_{d_C}\rangle = \sum_{a_B,b_A,b_C,c_B}(-1)^{a_B + (b_Ab_C)^y + c_B}p(a_B,b_A,b_C,c_B|d_A,d_C),
\end{equation}
with $y \in \{0,1\}$, and their linear combinations
\begin{align}
    I^{14} &= \frac{1}{4}\sum_{d_A,d_C}\langle A_{d_A}B^0C_{d_C}\rangle\\
    J^{14} &= \frac{1}{4}\sum_{d_A,d_C}(-1)^{d_A+d_C}\langle A_{d_A}B^1C_{d_C}\rangle
\end{align}

The inequality satisfied by all bilocal models has the form:
\begin{equation}\label{i14}
   \mathcal{I}_{14} \equiv \sqrt{|I^{14}|} + \sqrt{|J^{14}|} \leq 1.
\end{equation}
Consider now the square network as on Fig. \ref{fig:IndependentSquare}, and behaviours satisfying
\begin{align}\label{squarecond1}
    p(A_D=D_A) = 1, &\qquad\qquad p(C_D = D_C) = 1. 
\end{align}
Such behaviour satisfies, similarly like analogous behaviours in Appendix \ref{app:lemma}:
\begin{align}\nonumber
    p(a_B,A_D=a_D,b_A,b_C,c_B,C_D=c_D|\ld,\lp,\gamma,\delta) &= \sum_{d_A,d_C}p(a_B,A_D=a_D,b_A,b_C,c_B,C_D=c_D,D_A = d_A, D_C = d_C|\ld,\lp,\gamma,\delta)\\ \label{eq:SqoneStep}
    &= p(a_B, A_D=D_A=a_D, b_A, b_C, c_B C_D=D_C=c_D|\ld,\lp,\gamma,\delta). 
\end{align}
Similarly to~\eqref{eq:SqoneStep}, we also have equalities:
\begin{align}\nonumber
    p(a_B,D_A=d_A,b_A,b_C,c_B,D_C=d_C|\ld,\lp,\gamma,\delta) &= \sum_{a_D,c_D}p(a_B,A_D=a_D,b_A,b_C,c_B,C_D=c_D,D_A = d_A, D_C = d_C|\ld,\lp,\gamma,\delta)\\ \label{eq:SqtwoStep}
    &= p(a_B, A_D=D_A=a_D, b_A, b_C, c_B C_D=D_C=c_D|\ld,\lp,\gamma,\delta).
\end{align}
These equalities imply
\begin{multline}\label{SqidentityMoreGeneral}
    p(a_B,D_A=d_A,b_A,b_C,c_B,D_C=d_C|\ld,\lp,\gamma,\delta) = p(a_B,A_D=a_D,b_A,b_C,c_B,C_D=c_D|\ld,\lp,\gamma,\delta) 
    \\ \forall a_B,a_D,b_A,b_C,c_B,c_D,d_A,d_C,\ld,\lp,\gamma,\delta,
\end{multline}
and in particular
\begin{equation}\label{Sqidentity1}
    p(A_D=a_D,C_D=c_D|\ld,\lp,\gamma,\delta) = p(D_A=a_D,D_C=c_D|\ld,\lp,\gamma,\delta) 
    \qquad \forall a_D,c_D,\ld,\lp,\gamma,\delta.
\end{equation}\\
For a behaviour compatible with a quadrilocal hidden variable model and satisfying Eq. \eqref{squarecond1} we can write
\begin{align}\nonumber
    p(a_B,b_A,b_C,c_B,A_D = a_D, C_D = c_D|\alpha,\delta) &= \sum_{\lp,\la}p(\lp)p(\la)p(a_D,c_D|\ld,\lp,\la,\delta)p(a_B|a_D,\la,\delta)p(b_A,b_C|\delta,\ld)p(c_B|c_D,\alpha,\beta)\\
    &= \sum_{\lp,\la}p(A_D = a_D,C_D = c_D,\gamma,\beta|\alpha,\delta)p(a_B|a_D,\la,\delta)p(b_A,b_C|\delta,\ld)p(c_B|c_D,\alpha,\beta)\\
    &= \sum_{\lp,\la}p(D_A = a_D,D_C= c_D,\gamma,\beta|\alpha,\delta)p(a_B|a_D,\la,\delta)p(b_A,b_C|\delta,\ld)p(c_B|c_D,\alpha,\beta).
\end{align}
The last equation is the consequence of \eqref{Sqidentity1}.
Without loss of generality we can assume that the square-local model is deterministic (and put all randomness in the distribution of the hidden variables). That means that certain values of the pair $(\lp,\la)$ "induce" the pair of outcomes $D=(a_D,c_D)$ for Daisy, i.e. for certain values $(\lp,\la)$ it either holds $p(D_A = a_D,D_C = c_D,\lp,\la|\ld,\delta) = p(\lp,\la|\ld,\delta)$ or $p(D_A = a_D,D_C = c_D,\lp,\la|\ld,\delta) = 0$. Let us denote all the values of the pair $(\lp,\la)$ that "induce" the pair $(a_D,c_D)$ with $\Omega_{a_D,c_D}$. Hence the last expression reduces to
\begin{align}
    p(a_B,b_A,b_C,c_B,A_D = a_D,C_D = c_D|\alpha,\delta) &= \sum_{(\lp,\la)\in \Omega_{a_D,c_D}}p(\gamma,\beta|\alpha,\delta)p(a_B|a_D,\la,\delta)p(b_A,b_C|\delta,\ld)p(c_B|c_D,\alpha,\beta)\\
    &= \sum_{(\lp,\la)\in \Omega_{a_D,c_D}}p(\gamma,\beta|\alpha,\delta)p(a_B,c_B|a_D,c_D,\la,\delta,\alpha,\beta)p(b_A,b_C|\delta,\ld)\\
    &= \sum_{(\lp,\la)\in \Omega_{a_D,c_D}}p(\gamma,\beta|\alpha,\delta)p(a_B,c_B|\la,\delta,\alpha,\beta)p(b_A,b_C|\delta,\ld)\\
    &= \sum_{(\lp,\la)\in \Omega_{a_D,c_D}}p(a_B,c_B,\la,\beta|\delta,\alpha)p(b_A,b_C|\delta,\ld)\\
    &= p(a_B,c_B,a_D,c_D|\delta,\alpha)p(b_A,b_C|\delta,\ld)\\
    &= p(a_B,a_D|\delta)p(b_A,b_C|\delta,\ld)p(c_B,c_D|\alpha)
\end{align}
To obtain the second line we used the causal structure of the square network. In the third line we used the fact that $(\lp,\la) \in \Omega_{a_D,c_D}$, hence the information about $a_D$ and $c_D$ is contained in the values of the hidden variables. The fifth equation is the consequence of \eqref{Sqidentity1}, and the final expression holds because of the causal structure of the square network. If we divide the last expression with $p(A_D = a_D,C_D = c_D|\alpha,\delta)=P(A_D=a_D|\alpha)P(C_D=c_D|\delta)$ and average over all values of $\alpha$ and $\delta$ we obtain the conditional probability distribution
\begin{align}
    p(a_B,b_A,b_C,c_B|a_D,c_D) &=  \sum_{\ld,\delta}p(\ld)p(\delta)p(a_B|a_D,\delta)p(b_A,b_C|\ld,\delta)p(c_B|c_D,\alpha)
\end{align}
which cannot violate any bilocality inequality. Hence the behaviour characterized by conditions \eqref{squarecond1} and violating inequality  \eqref{i14} is incompatible with local hidden variable models. 

\subsection{Nonlocality in square network with partially independent sources}\label{app:overlappingSquare}

In this appendix we show how to detect nonlocality in the square network with partially independent sources. The scenario is represented on Fig. \ref{fig:OverlappingSquare}. The central hidden variable $\Lambda$ is taking values ${\lambda}$. A behaviour $p(a,b,c,d)$ is square-local if it admits the following decomposition
\begin{equation}\label{squareDSlocalSM}
  p(a,b,c,d) = \sum_{{\lambda},\alpha,\beta,\gamma,\delta} p({\lambda})p(\alpha,\beta,\gamma,\delta|{\lambda})  p(a|\gamma,\delta)p(b|\alpha,\delta)p(c|\alpha,\beta)p(d|\beta,\gamma).
\end{equation}
The partial independence of sources eqs. \eqref{sqind1},\eqref{sqind2} imply
\begin{align}\nonumber
 p(\beta)p(\gamma)\varepsilon_1 &\leq  p(\beta,\gamma|{\lambda}) \leq p(\beta)p(\gamma)\varepsilon_2\qquad \forall \lp,\la,\lt,\\ \label{SqindependenceCondition2}
 p(\lp)\varepsilon_1 &\leq p(\beta|\alpha,\lt) \leq p(\lp)\varepsilon_2(\ld,\lp,\lt)\qquad \forall \alpha,\beta,\lt\\ \nonumber
 p(\la)\varepsilon_1 &\leq p(\la|\delta,\lt) \leq p(\la)\varepsilon_2(\la,\delta,\lt)\qquad \forall \gamma,\delta,\lt\\ \nonumber
 p(\ld)p(\delta)\varepsilon_1 &\leq p(\ld,\delta)
 \end{align}
where $\varepsilon_2 = \max_{\alpha,\beta,\gamma,\delta,\lambda}\varepsilon_2(\alpha,\beta,\gamma,\delta,\lambda)$.
The relations \eqref{SqindependenceCondition2} allow to derive the conditions for a square behaviour to be nonlocal. For that purpose let us consider the behaviour from the appendix \ref{app:independentSquare}:
\begin{equation}\label{SqfritzdistributionGen}
    p(A_D = D_A \wedge C_D = D_C) = 1, \qquad \mathcal{I}_{14} > 1,
\end{equation}
with $p(A_D)>0$ and $p(C_D)>0$, i.e. when all outcomes of $A_D$ and $C_D$ can be observed.
The first relation in \eqref{SqfritzdistributionGen} implies (the same way as we proved Eq. \eqref{identity}) 
\begin{equation}\label{Sqidentity}
    p(A_D = a_D, C_D = c_D|\ld,\lp,\la,\delta,{\lambda}) = p(D_A = a_D,D_C = c_D|\ld,\lp,\la,\delta,{\lambda}).
\end{equation}
We proceed by writing  the behaviour $p(a_B,A_D = a_D,b_A,b_C,c_B,C_D = c_D)$ compatible with local hidden variable models: 
\begin{align}
    p(a_B,A_D = a_D,b_A,b_C,&c_B,C_D = c_D) =  \sum_{\substack{\lt,\ld,\lp,\\\la,\delta}} p(\lt)p(\ld,\delta|\lt)p(\lp,\la|\ld,\delta,\lt)p(a_B,A_D = a_D,c_B,C_D = c_D|\ld,\lp,\la,\delta,\lt)p(b_A,b_C|\ld,\delta)\\
    &=   \sum_{\substack{\lt,\ld,\lp,\\\la,\delta}} p(\lt)p(\ld,\delta|\lt)p(\lp,\la|\ld,\delta,\lt)p(a_B,D_A = a_D,c_B,D_C = c_D|\ld,\lp,\la,\delta,\lt)p(b_A,b_C|\ld,\delta)\\ 
    &=  \sum_{\lt,\ld,\lp,\la,\delta} p(\lt)p(\ld,\delta|\lt)p(D_A=a_D,D_C=c_D,\lp,\la|\ld,\delta,\lt)p(a_B|\la,\delta,\lt)p(b_A,b_C|\ld,\delta,\lt)p(c_B|\ld,\lp,\lt)\\
    &=  \sum_{\lt,\ld,\lp,\la,\delta} p(\lt)p(\ld,\delta|\lt)p(D_A=a_D,D_C=c_D,\lp,\la|\lt)p(a_B|\la,\delta,\lt)p(b_A,b_C|\ld,\delta,\lt)p(c_B|\ld,\lp,\lt)
\end{align}
 To get the last equation we used the fact that value of $D$ and $\beta,\gamma$, conditioned on $\lambda$ do not depend on $\alpha$ and $\delta$.
 For each value of  $D$ in the last equation the only contributions in the sums over $\lp$ and $\la$ come from the values of $(\lp,\la) \in \Omega_{a_D,c_D}$. Let us now define  $\eta_2' = \max_{\ld,\lp,\lt}\eta_2'(\ld,\lp,\lt)$ and $\eta_2'' = \max_{\la,\delta,\lt}\eta_2''(\la,\delta,\lt)$. Hence, the last expression reduces to
\begin{align}
    p(a_B,A_D = a_D,b_A,b_C,c_B,C_D = c_D) &=  \sum_{\substack{\lt,\ld,\delta,\\(\lp,\la)\in \Omega_{a_D,c_D}}} p(\lt)p(\ld,\delta|\lt)p(\lp,\la|\lt)p(a_B|\la,\delta,\lt)p(b_A,b_C|\ld,\delta,\lt)p(c_B|\ld,\lp,\lt)
\end{align} 
Given \eqref{SqindependenceCondition2} we can put upper and lower bounds on this expression:
\begin{align}
p(a_B,A_D = a_D,&b_A,b_C,c_B,C_D = c_D) \quad \geq \quad
\varepsilon_1 \sum_{\substack{\lt,\ld,\delta\\
\lp,\la \in \Omega_{a_D,c_D}}}p(\lt)p(\ld,\delta|\lt)p(\lp)p(\la)p(a_B|\la,\delta,\lt)p(b_A,b_C|\ld,\delta,\lt)p(c_B|\ld,\lp,\lt)\\
 &\quad \geq \quad
\frac{\varepsilon_1}{\varepsilon_2^2} \sum_{\substack{\lt,\ld,\delta\\
\lp,\la \in \Omega_{a_D,c_D}}}p(\lt)p(\ld,\delta|\lt)p(\lp|\ld,\lt)p(\la|\delta,\lt)p(a_B|\la,\delta,\lt)p(b_A,b_C|\ld,\delta,\lt)p(c_B|\ld,\lp,\lt)\\
 &\quad = \quad
\frac{\varepsilon_1}{\varepsilon_2^2} \sum_{\substack{\lt,\ld,\delta\\
\lp,\la \in \Omega_{a_D,c_D}}}p(\lt)p(\ld,\delta|\lt)p(a_B,\la|\delta,\lt)p(b_A,b_C|\ld,\delta,\lt)p(c_B,\lp|\ld,\lt)
\\
 &\quad = \quad
\frac{\varepsilon_1}{\varepsilon_2^2} \sum_{\substack{\lt,\ld,\delta\\
\lp,\la \in \Omega_{a_D,c_D}}}p(\lt)p(\ld,\delta|\lt)p(a_B,c_B,\la,\lp|\ld,\delta,\lt)p(b_A,b_C|\ld,\delta,\lt)
\\
 &\quad = \quad
\frac{\varepsilon_1}{\varepsilon_2^2} \sum_{\substack{\lt,\ld,\delta\\
\lp,\la \in \Omega_{a_D,c_D}}}p(\lt)p(\ld,\delta|\lt)p(a_B,c_B,D_A = a_D,D_C = c_D|\ld,\delta,\lt)p(b_A,b_C|\ld,\delta,\lt)
\\ &\quad = \quad
\frac{\varepsilon_1}{\varepsilon_2^2} \sum_{\substack{\lt,\ld,\delta\\
\lp,\la \in \Omega_{a_D,c_D}}}p(\lt)p(\ld,\delta|\lt)p(a_B,c_B,A_D = a_D,C_D = c_D|\ld,\delta,\lt)p(b_A,b_C|\ld,\delta,\lt)
\\
 &\quad = \quad
\frac{\varepsilon_1}{\varepsilon_2^2} \sum_{\lt,\ld,\delta}p(\lt)p(\ld,\delta|\lt)p(a_B,A_D = a_D|\delta)p(b_A,b_C|\ld,\delta)p(c_B,C_D = c_D|\ld)
\\
 &\quad \geq \quad
\frac{\varepsilon_1^2}{\varepsilon_2^2} \sum_{\ld,\delta}p(\ld)p(\delta)p(a_B,a_D|\delta)p(b_A,b_C|\ld,\delta)p(c_B,c_D|\ld)\\
 &\quad = \quad
\zeta_1 \sum_{\ld,\delta}p(\ld)p(\delta)p(a_B,a_D|\delta)p(b_A,b_C|\ld,\delta)p(c_B,c_D|\ld) \label{Sqquarantine1}
\end{align}
where $\zeta_1 = \frac{\varepsilon_1^2}{\varepsilon_2^2}$.
The last relation implies that there is some $\zeta_2 < 1/p(a_B,a_D,b_A,b_C,c_B,c_D)$, such that through the normalization of probabilities it must be
\begin{equation}\label{Sqquarantine2}
    p(a_B,a_D,b_A,b_C,c_B,c_D) \leq \zeta_2 \sum_{\ld,\delta}p(\ld)p(\delta)p(a_B,a_D|\delta)p(b_A,b_C|\ld,\delta)p(c_B,c_D|\ld),
\end{equation}
for all $a_B,a_D,b_A,b_C,c_B,c_D$. Through the same chain of arguments one can show that $\zeta_2 = \frac{\varepsilon_2^2}{\varepsilon_1^2}$.
Let us now further develop the expression \eqref{Sqquarantine1} and \eqref{Sqquarantine2}: to obtain
\begin{align}\label{Sqselfisolate1}
    p(a_B,a_D,b_A,b_C,c_B,c_D) &\quad \geq \quad
\zeta_1 \sum_{\ld,\delta}p(\ld)p(\delta)p(a_D|\delta)p(a_B|a_D,\delta)p(b_A,b_C|\ld,\delta)p(c_D|\ld)p(c_B|c_D,\ld)\\ \label{Sqselfisolate2}
p(a_B,a_D,b_A,b_C,c_B,c_D) &\quad \leq \quad
\zeta_2 \sum_{\ld,\delta}p(\ld)p(\delta)p(a_D|\delta)p(a_B|a_D,\delta)p(b_A,b_C|\ld,\delta)p(c_D|\ld)p(c_B|c_D,\ld)
\end{align}

Observe now the following bounds
\begin{align}\label{1100}
    \varepsilon_1p(\beta)p(\gamma)p(\delta) \qleq &p(\beta,\gamma,\delta) \qleq \varepsilon_2 p(\beta)p(\gamma)p(\delta)\\ \label{1101}
    \varepsilon_1p(\beta)p(\gamma) \qleq &p(\beta,\gamma) \qleq \varepsilon_2 p(\beta)p(\gamma)\\ \label{1110}
    \varepsilon_1p(\beta)p(\gamma) \qleq &p(\beta,\gamma|\delta) \qleq \varepsilon_2 p(\beta)p(\gamma),
\end{align}
where $\varepsilon_2 = \max_{\alpha,\beta,\gamma,\delta,\lambda}\varepsilon_2(\alpha,\beta,\gamma,\delta,\lambda)$. Inequalities \eqref{1100} and \eqref{1101} are obtained from relations \eqref{sqind1} and \eqref{sqind2} by multiplying with $\lambda$ and summing over over $\lambda$ and $\alpha$ (\eqref{1100}), or $\lambda$, $\alpha$ and $\delta$ (\eqref{1101}). Inequality \eqref{1110} is obtained from \eqref{1100} by simply dividing with $p(\delta)$. Inequalities \eqref{1101} and $\eqref{1110}$ imply:
\begin{equation}\label{1200}
    \frac{\varepsilon_1}{\varepsilon_2}p(\beta,\gamma) \qleq p(\beta,\gamma|\delta) \qleq \frac{\varepsilon_2}{\varepsilon_1}p(\beta,\gamma)
\end{equation}
In the same way we can obtain the following inequality:
\begin{equation}\label{1201}
    \frac{\varepsilon_1}{\varepsilon_2}p(\beta,\gamma) \qleq p(\beta,\gamma|\alpha) \qleq \frac{\varepsilon_2}{\varepsilon_1}p(\beta,\gamma)
\end{equation}
Let us now define the set $\Omega_{a_D}$ to be the set of all pairs $(\beta,\gamma)$, which "induce" $D_A = a_D$, and similarly $\Omega_{c_D}$ to be the set of all pairs $(\beta,\gamma)$, which "induce" $D_C = c_D$.
Taking into account that $p(A_D = a_D|\delta) = p(D_A = a_D|\delta) = \sum_{\beta,\gamma \in \Omega_{a_D}}p(\gamma,\beta|\alpha)$ and $p(C_D = c_D|\delta) = p(D_C=c_D|\alpha) = \sum_{\beta \in \Omega_{c_D}}p(\beta,\gamma|\alpha)$ we can obtain the following relations, by summing ineqs. \eqref{1200} and $\eqref{1201}$ over all pairs $(\beta,\gamma)$ belonging to $\Omega_{a_D}$ and $\Omega_{c_D}$, respectively
\begin{align}\label{Sqthetaa}
    \theta_{a,1} \qleq \frac{\varepsilon_1}{\varepsilon_2}p(A_D = a_D) \leq p(A_D = a_D|\delta) \leq \frac{\varepsilon_2}{\varepsilon_1}p(a_D) \qleq \theta_{a,2} \qquad \forall\delta \\ \label{Sqthetab}
    \theta_{c,1} \qleq \frac{\varepsilon_1}{\varepsilon_2}p(C_D = c_D) \leq p(C_D = c_D|\alpha) \leq \frac{\varepsilon_2}{\varepsilon_1}p(C_D = c_D) \qleq \theta_{c,2} \qquad \forall\alpha,
\end{align}
where $\theta_{a,1} = \frac{\varepsilon_1}{\varepsilon_2}\min_{a_D}p(a_D)$, $\theta_{a,2} = \frac{\varepsilon_2}{\varepsilon_1}\max_{a_D}p(a_D)$, $\theta_{c,1} = \frac{\varepsilon_1}{\varepsilon_2}\min_{c_D}p(c_D)$, $\theta_{c,2} = \frac{\varepsilon_2}{\varepsilon_1}\max_{c_D}p(c_D)$. 
Hence, by combining eqs. \eqref{Sqselfisolate1}, \eqref{Sqselfisolate2}, \eqref{Sqthetaa} and \eqref{Sqthetab} we obtain
\begin{align}\label{Sqselfisolate11}
    p(a_B,a_D,b_A,b_C,c_B,c_D) &\quad \geq \quad
\zeta_1\theta_{a,1}\theta_{c,1} \sum_{\ld,\delta}p(\ld)p(\delta)p(a_B|a_D,\delta)p(b_A,b_C|\ld,\delta)p(c_B|c_D,\ld)\\ \label{Sqselfisolate22}
p(a_B,a_D,b_A,b_C,c_B,c_D) &\quad \leq \quad
\zeta_2\theta_{a,2}\theta_{c,2} \sum_{\ld,\delta}p(\ld)p(\delta)p(a_B|a_D,\delta)p(b_A,b_C|\ld,\delta)p(c_B|c_D,\ld)
\end{align}
Before presenting the final results let us introduce $\xi_1 = \zeta_1\theta_{a,1}\theta_{c,1}$ and $\xi_2 = \zeta_2\theta_{a,2}\theta_{c,2}$.

 In a similar manner as in Appendix \ref{app:lemma} we can build a source-dependent locality bound. The expression $\mathcal{I}_{14}$, defined in \eqref{i14} has the following form:
 \begin{multline}
     \mathcal{I}_{14}(p(a_B,b_A,b_C,c_B|a_D,c_D)) = \sqrt{\left|\sum_{a,b,c}\left(w_{a,b,c}^+p(a_B,b_A,b_C,c_B|a_D,c_D) - w_{a,b,c}^-p(a_B,b_A,b_C,c_B|a_D,c_D)\right)\right|} +\\+ \sqrt{\left|\sum_{a,b,c}\left({w'}_{a,b,c}^+p(a_B,b_A,b_C,c_B|a_D,c_D) - {w'}_{a,b,c}^-p(a_B,b_A,b_C,c_B|a_D,c_D)\right)\right|} ,
 \end{multline}
where all $w_{a,b,c}^+$,$w_{a,b,c}^-$, ${w'}_{a,b,c}^+$ and ${w'}_{a,b,c}^-$ are nonnegative. Let us now define the following source-dependent Bell expression:
\begin{multline}
    I \equiv \mathcal{I}^{sd}_{14}(p(a_B,b_A,b_C,c_B,a_D,c_D)) = \sqrt{\left|\sum_{a,b,c}\left(\xi_1w_{a,b,c}^+p(a_B,b_A,b_C,c_B|a_D,c_D) - \xi_2w_{a,b,c}^-p(a_B,b_A,b_C,c_B|a_D,c_D)\right)\right|} +\\+ \sqrt{\left|\sum_{a,b,c}\left(\xi_1{w'}_{a,b,c}^+p(a_B,b_A,b_C,c_B|a_D,c_D) - \xi_2{w'}_{a,b,c}^-p(a_B,b_A,b_C,c_B|a_D,c_D)\right)\right|}.
 \end{multline}
Eqs. \eqref{Sqselfisolate11} and \eqref{Sqselfisolate22} then imply
\begin{align}\nonumber
    I &\leq \sqrt{\xi_1\xi_2}
\end{align}
The above given bound is satisfied for all local models with overlapping sources, satisfying the condition \eqref{sqind1}.

\end{document}